\documentclass[10pts,conference]{IEEEtran}
\makeatletter
\def\ps@headings{%
\def\@oddhead{\mbox{}\scriptsize\rightmark \hfil \thepage}%
\def\@evenhead{\scriptsize\thepage \hfil \leftmark\mbox{}}%
\def\@oddfoot{}%
\def\@evenfoot{}}
\makeatother
\usepackage{graphicx}
\usepackage{amssymb,color}
\usepackage[cmex10]{amsmath}
\usepackage{algorithm}
\usepackage[font=small]{caption}
\usepackage{bbm,bm}\usepackage{comment,subfigure}

\newtheorem{cor}{Corollary}

\newtheorem{lemma}{Lemma}

\newtheorem{thm}{Theorem}
\newtheorem{remark}{Remark}
\newcommand{\fdp}[1]{\textcolor{black}{#1}}
\newcommand{\dm}[1]{\textcolor{black}{#1}}

\def\Xps{X_{ps}}

\def\Xpu{X_{pu}}
\def\Lps{\lambda_{ps}}
\def\Lpu{\lambda_{pu}}
\def\L{\lambda}

\newcommand{\bpi}{{\bm{\pi}}}
\newcommand{\balpha}{\bm{\alpha}}

\newcommand{\One}{\mathbbm{1}}
\usepackage{flushend}

\IEEEoverridecommandlockouts
\begin{document}

\title{Emergence of Equilibria from Individual Strategies in Online Content Diffusion}
\author{Eitan Altman$^\star$\thanks{$^\star$INRIA B.P.93, 2004 Route des Lucioles, 06902 Sophia-Antipolis, Cedex, FRANCE, email: {\tt\small eitan.altman@sophia.inria.fr}}, Francesco De Pellegrini$^{\diamond}$\thanks{$^{\diamond}$CREATE-NET, via Alla Cascata 56 c, 38100 Trento, ITALY, email: {\tt\small \{fdepellegrini,dmiorandi\}@create-net.org}}, Rachid El-Azouzi$^\dagger$\thanks{$^\dagger$CERI/LIA, University of Avignon, 74 rue Louis Pasteur, 84029 AVIGNON Cedex 1, email: {\tt\small \{rachid.elazouzi, tania.altman\}@univ-avignon.fr}}, Daniele Miorandi$^{\diamond}$ and Tania Jimenez$^\dagger$}

\maketitle
\begin{abstract}
Social scientists have observed that human behavior in society can often
be modeled as corresponding to a threshold type policy. A new behavior
would propagate by a procedure in which an individual adopts the new
behavior if the fraction of his neighbors or friends having adopted
the new behavior exceeds some threshold. In this paper we study 
the question of whether the emergence of threshold policies may be
modeled as a result of some rational process which would describe the behavior
of non-cooperative rational members of some social network.  We focus
on situations in which individuals take the decision whether to access
or not some content, based on the \fdp{number of views that the} content has.
Our analysis aims at understanding not only the behavior of individuals,
but also the way in which information about the quality of a given content can be
deduced from view counts when only part of the viewers that
access the content are informed about its quality.
In this paper we present a game formulation for the behavior
of individuals using a meanfield model: the number of
individuals is approximated by a continuum of atomless players
and for which the Wardrop equilibrium is the solution concept.
We derive conditions on the problem's parameters that result indeed
in the emergence of threshold equilibria policies. But we also
identify some parameters in which other structures are obtained
for the equilibrium behavior of individuals.
\end{abstract}
\begin{IEEEkeywords}
User-generated content, Complex Systems, Video popularity, Game theory, Wardrop equilibria
\end{IEEEkeywords}


\section{Introduction}


Online media constitute currently the largest share of Internet
traffic. A large part of such traffic is generated by platforms that
deliver user-generated content (UGC). This includes, among the other
ones, YouTube and Vimeo for videos, Flickr and Instagram for images 
and all social networking platforms.

Among such services, a prominent role is played by YouTube. Founded in $2005$ 
by Chad Hurley, Steve Chen and Jawed Karim and acquired in $2006$ by Google, 
YouTube scored in $2011$ more than $1$ trillion views (or, alternatively, 
an average of $140$ video views for every person on Earth), with more 
than $3$ billion hours of video watched every month and $72$ hours 
of video uploaded every minute by YouTube's users\footnote{{\tt http://www.youtube.com/t/press\_statistics/}}. 

Of course, not all videos posted on YouTube are equal. The key aspect
is their ``popularity'', broadly defined as the number of views they
score (also referred to as {\em viewcount}). This is relevant from a 
twofold perspective. On the one hand, more popular content generates more traffic, so understanding
popularity has a direct impact on caching and replication strategy
that the provider should adopt. On the other one, popularity has a direct economic impact. Indeed,
popularity or viewcount are often directly related to click-through
rates of linked advertisements, which constitute the basis of the
YouTube's business model.

Recently, a number of researchers have analysed the evolution of the
popularity of online media content~\cite{ChaUtube,crane2008viral,Gill07youtubetraffic,RatkiewiczBurstyPoP,ChatFirstStep,ChaTON},
with the aim of developing models for early-stage prediction of 
future popularity~\cite{SzaboPop}. 

Such studies have highlighted a number of phenomena that are typical of UGC
delivery. This includes the fact that a significant share of content
gets basically no views~\cite{ChaTON}, as well as the fact that popularity may see some
bursts, when content ``goes viral''~\cite{RatkiewiczBurstyPoP}. 
Also, in~\cite{SzaboPop}  the authors demonstrate that after an
initial phase, in which contents gain popularity through
advertisement and other marketing tools, the platform mechanisms to 
induce users to access contents (re-ranking mechanisms)
are main drivers of popularity.

In this paper, we address such phenomena, by developing a model, based
on game theoretical concepts and tools, for understanding how user's
behaviour drives the evolution of popularity of a given content. The
work is based on rational decision-making assumptions, whereby the
users have to decide whether to see a given content or not. This configures
as a game, where users seek to maximize 
some expected utility based on their ``perception'' of the quality of
the content\footnote{This may come, e.g., from the name of user
  who posted the content.} 
and on viewcount. However, users suffer 
also a cost for accessing contents of bad quality, i.e., waste of 
time and possibly bandwidth, batteries, etc. In particular, in the 
decision process the viewcount is used as a noisy estimator of 
the quality of a content. Interestingly, this context resembles 
closely the situation in the economic domain, where  customers of 
a firm which are uninformed do infer the quality of products 
from the length of the queue they encounter upon requesting 
firm's goods to purchase~\cite{Debo2012}. 

Extensive advertising and marketing campaigns can be used to push the 
viewcount of a given content up. And in the decision making process 
users do not know whether the viewcount has been ``pushed'' by such 
means. Also, the decisions made by different users influence the 
viewcount and consequently the decisions made by other users, a 
process which suits well the usage of game theoretical machinery. 

Specifically, we describe the conditions for the adoption of common 
behaviors in online content access. This is inspired by findings
in social science \cite{RolfeSocNet,GraSoongJEBO1986,GranovetterAJS1978}:
results there show that emerging behaviours would propagate by a 
procedure in which an individual adopts a novel behavior if the 
fraction of neighbors or friends having adopted the same behavior 
exceeds some threshold. In our context, the threshold would 
be expressed in terms of viewcount or related metric.

In the sense of game theory, users of online media represent
non-cooperative rational players connected through some social tie, 
e.g., being users of the same UGC platform. Since we consider systems composed by a very large number of users, 
the customary tool  to study the user behaviour is that of Wardrop 
equilibria~\cite{wardrop52}. In particular, we have found a number 
of conditions for which such equilibria exist and can be characterized 
analytically. Explicit conditions were found for content to stay at zero views or to become
so popular that it is makes sense for all users to access it the sooner the 
better.

Furthermore, we identify, for the general case, conditions under which players 
tend to accrue around a common strategy depending on initial conditions. This is due to the 
existence  of a continuum of equilibria: the system will settle at any point very much depending on initial conditions imposed, for instance, by a set of forerunners which cause significant changes of the content popularity. Such conditions were identified in early works such as \cite{Hassin97equilibriumthreshold} in other contexts: there, the authors applied threshold type Nash equilibrium strategies in which one purchases priority if and only if upon arrival the queue size is larger than some threshold value. Key motivation in \cite{Hassin97equilibriumthreshold} is predictability and control of purchase priority. What motivates this work is predictability and control of online content access.

\paragraph*{Novel contribution} 
in this paper, we move away from the classical analysis of social networks in the spirit of \cite{SzaboPop,RatkiewiczBurstyPoP,ChatFirstStep,ChaUtube}: instead, we provide a first analysis based on games. The aim of this paper is to provide a novel perspective where contents compete to gain popularity and are subject to the effect of user's choice. To the best of the authors' knowledge, this is the first attempt so far to describe content popularity in UGC systems using game theoretical tools. 

The remainder of the work is organized as follows. In Sec.~\ref{sec:model} we introduce the system model and the notation
used throughout the paper. Results for the case when plain viewcount is
used to make decisions are presented in Sec.~\ref{sec:gt1}. When
decisions account also \fdp{for a large increasing trend} of content popularity, i.e., looking
for 'hot' content, the dynamics of the game becomes different. This
case is analysed in Sec.~\ref{sec:gt2}. In Sec.~\ref{sec:gt3} we analyze the joint effect \fdp{when 
both the viewcount and its trend are both relevant to the user}. Finally, in Sec.~\ref{sec:sideinfo} we model the 
effect of side information when users have some measure of future content dynamics. 
Sec.~\ref{sec:rel} reviews the related work and Sec.~\ref{sec:concl} concludes the paper 
highlighting directions for expanding the current reach of the work.

\begin{figure}[t]
  \centering
  \subfigure[``President Obama Sings Sweet Home Chicago"]{\includegraphics[width=0.4\textwidth]{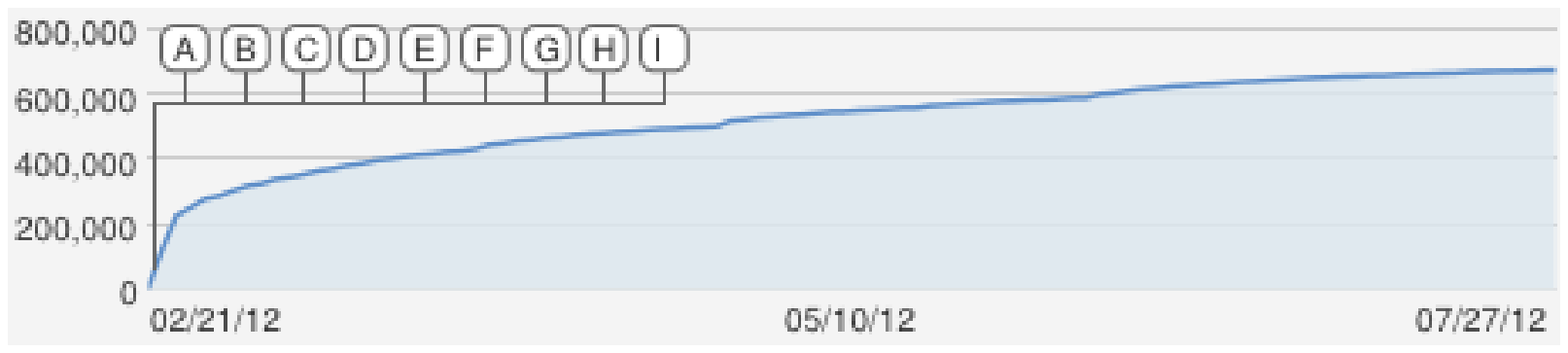}}
  \subfigure[``Chris Sharma Worlds' First 5.15'']{\includegraphics[width=0.4\textwidth]{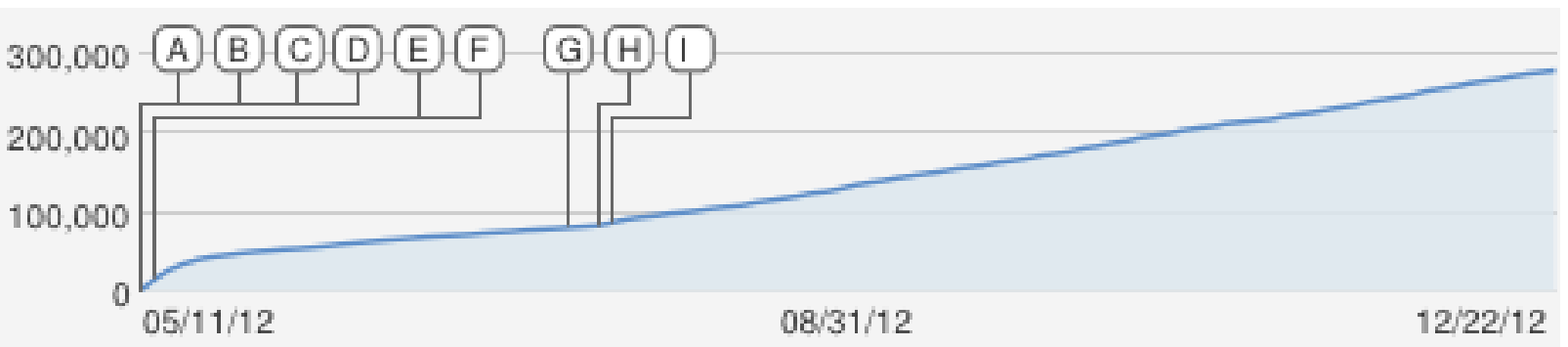}}
  \subfigure[``Montersino's Sacher Cake"]{\includegraphics[width=0.4\textwidth]{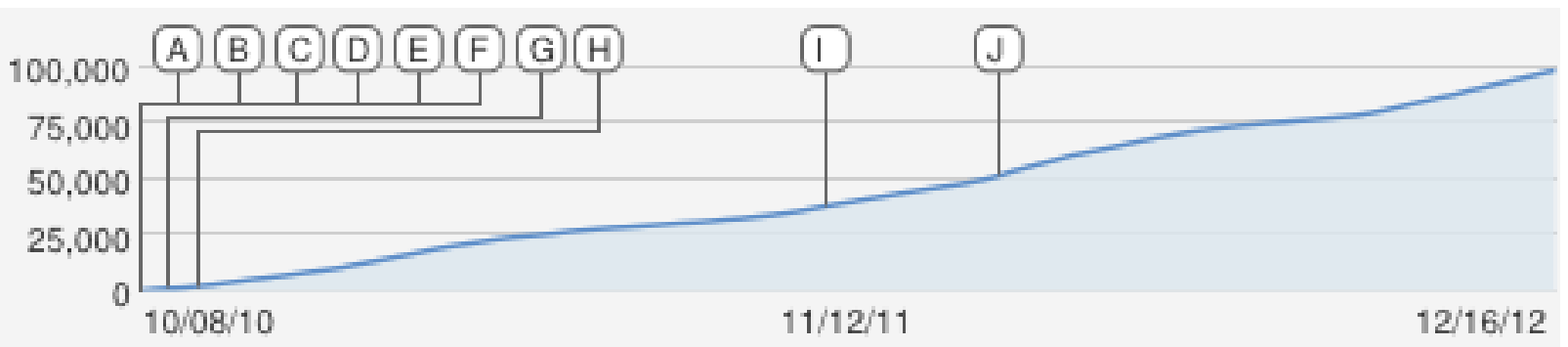}}
   \subfigure[``Shakira -- Waka-Waka"]{\includegraphics[width=0.4\textwidth]{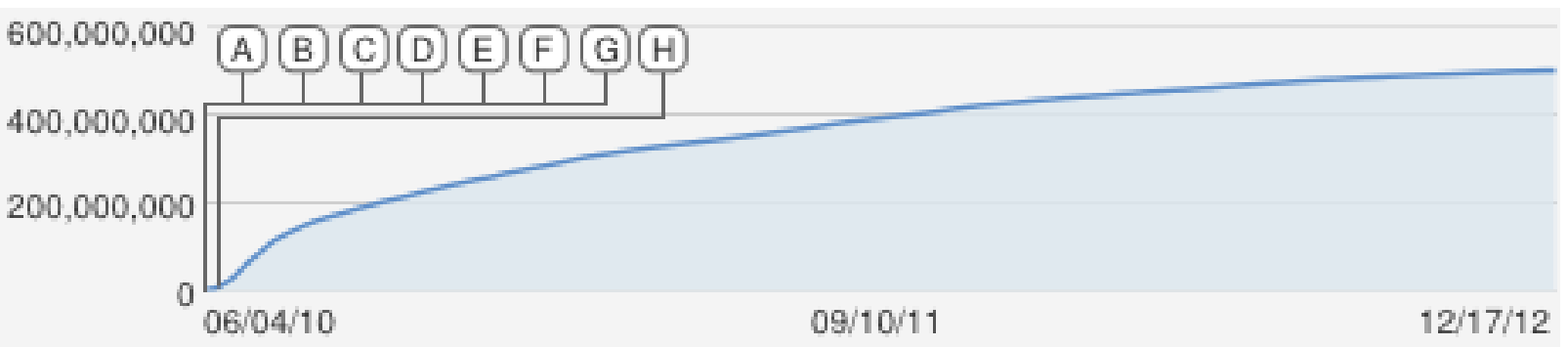}}
  \subfigure[``Bruno Mars -- Grenade"]{\includegraphics[width=0.4\textwidth]{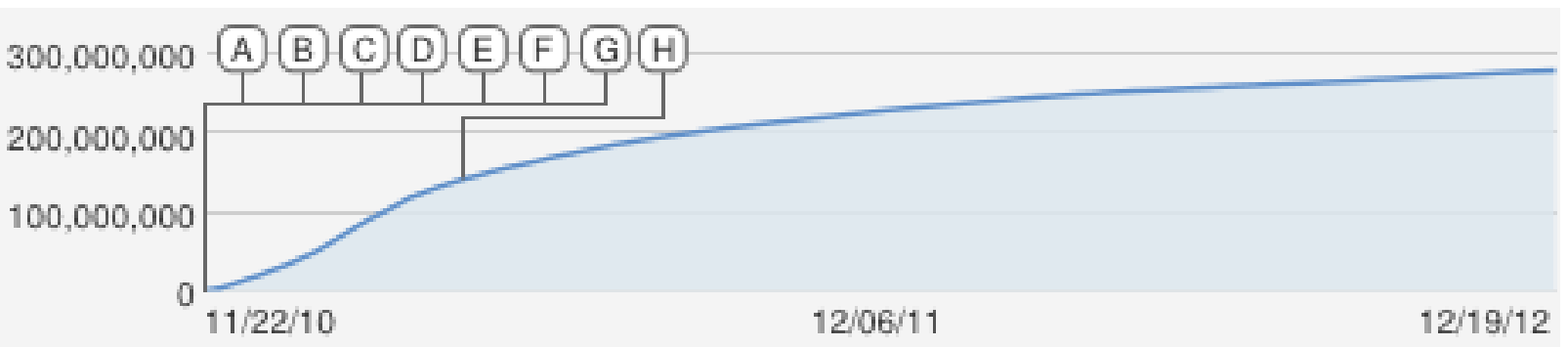}}
  \subfigure[``Adele -- Rolling in the deep"]{\includegraphics[width=0.4\textwidth]{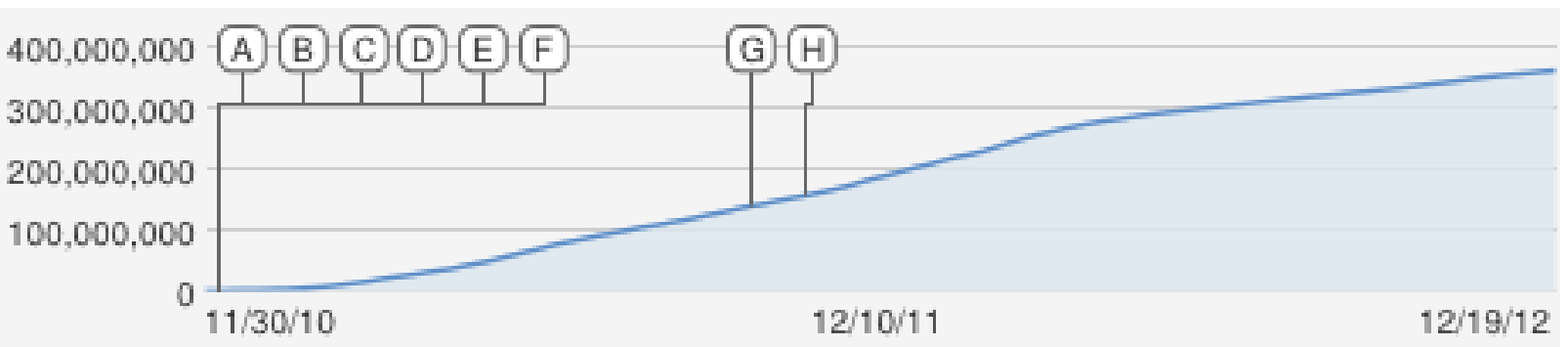}}
  \caption{\fdp{Dynamics of the viewcount for six sample videos: the push dynamics can be identified with the first part of the dynamics, where labels identify some actions that are significant for the diffusion of the video; observe for  cases a, b and c how a linear dynamics takes over in the last part of the dynamics. The labels tagging the first part of the dynamics mention specific events that identify the diffusion of the content on specific platforms or channels.}}\label{fig:pictures}
\end{figure}


\section{System Model}\label{sec:model}



We consider contents made available to a user by means of YouTube or a similar platform. We denote by $\tau$ the lifetime of a content, i.e.,
the time horizon during which the content bears some interest. In general, such horizon  differs depending on the type of content: it can be typically 
of the order of weeks to months for YouTube videos or a few days for news \cite{SzaboPop}. \fdp{A possible extension to the case of  variable time horizon is the addressed in Sec.~\ref{sec:gt2}.}

We denote by $X(t)$ the viewcount attained by a given content $\theta$ at time $t$ seconds after it has been posted, for $0\leq t \leq \tau$. 

As in standard UGC platforms, there are two mechanisms that coexist and \fdp{can jointly increase the viewcount:}
\begin{itemize}
\item {\em push}: the content provider exploits some preferential channels (including paid advertisement either directly on the UGC system or via social networking platforms) to make users aware of the content and to induce them to access it. We call {\em push users} the users that access the content as a reaction to the push mechanism. 
\item {\em pull}: users find about the content through standard search and decide to access it based on the belief that the content is relevant for them. We call users accessing a content through the pull mechanism {\em pull users}.
\end{itemize}


\fdp{In practice, many YouTube videos are subject to the push and the pull mechanisms described above such as 
the examples  that  we reported in Fig.~\ref{fig:pictures}. For instance, Fig.~\ref{fig:pictures}a, shows the dynamics of a popular video with viewcount $X \geq 675000$. The YouTube statistics associated with the video describe explicitly a series of events happening in the first part of the dynamics of $X$. For instance, the event B that appears around 02/12/2012, is precisely the event \texttt{``First embedded on: plus.google.com''} which indeed configures as a push towards a social network platform. After the initial push, such events vanish, and the rest of the dynamics appears ascribed mostly to the pull mechanism defined above, with a linear increase in the viewcount.}

\fdp{Also, some of the reported videos are representative of a specific class of online contents, which are those we will be dealing with in the rest of the paper. We can refer to those as the contents that comply to the {\em exponential-linear} model, for the sake of brevity. In particular, many such contents appear to obey to the following dynamics: after an initial exponential growth, the increase of the viewcount becomes linear. The way to interpret such a behavior can be traced to the notion of push and pull mechanisms described above: the exponential growth corresponds to actions through which the source distributes the content within a basin of target push viewers. When such basin is finite and small with respect to the content diffusion dynamics, the viewcount dynamics experiences a 
saturation effect which takes over after an initial phase. However, at that stage, the access to the content is due to pull users that come across the content browsing online: they do so at random from a very large basin, so that the access rate, i.e., the viewcount increase rate, is linear.  These combined effects are visible in the case of the first two videos, i.e. Fig~\ref{fig:pictures}a and Fig~\ref{fig:pictures}b. In the case of the first video, the saturation effect is well visible, whereas in the case of the second one the linear increase following the saturation is dominating. The example in Fig~\ref{fig:pictures}c is a case where all the dynamics is linear with good approximation: as it will be clear in the following, in the exponential-linear model this case is represented when either the basin of push users is large or when the rate at which contents are pushed is small. 
\begin{remark}
Not all videos will diffuse according to the proposed exponential-linear model. For instance,  there exist cases when the initial viewcount dynamics displays a characteristic sigmoid shape. We reported in Fig~\ref{fig:pictures}d,e,f the viewcount dynamics for three popular music videos: in those cases the dynamics resembles the logistic curve associated to the spread of epidemics. We can ascribe such similarity to the presence of a positive feedback in the push mechanism, e.g., those who access the content have some mean to recommend the content for others to access it, through targeted recommendation or similar mechanism. When a social network is present, this may happen due to the push of the content into the neighborhood of those who view the content. A similar and perhaps more powerful feedback effect can happen between different channels on the same platform, e.g., YouTube channels, and across different platforms through the recommendation list that is presented to the platform users.\\
This also qualifies the type of exponential-linear dynamics that we consider as those for which this type of feedback does not play a significant role. In particular, in the case of Fig~\ref{fig:pictures}a, the content is of interest at the national scale in the US, and the viewers are likely driven 
to the content by general search criteria (e.g., typing in a search engine). Also, in the case of Fig~\ref{fig:pictures}c, the viewers are likely those who browse for some specific recipe, whereas in the case of Fig~\ref{fig:pictures}b viewers are interested in a niche sport, where the event is known within the reference community. In all such cases we see that the linear part of the dynamics takes over and becomes dominant. 
\end{remark}
}

\subsection*{Game model} 

In our model, we are interested in the uptake of the pull users. Pull users interested in the given content do not know in advance its quality. They may discover it during interval $[0,\tau]$ at random. Their estimation of the interest/potential quality is based on the viewcount $X$. In the simplest case, contents with higher viewcount are more likely to be accessed. 

We define by $\Xps(t)$ the number of push users accessing the content up to time $t$ as a reaction to the push mechanism and, analogously, by $\Xpu(t)$ the number of those accessing it through the pull mechanism. Clearly, $X(t)=\Xps(t)+\Xpu(t)$.

Users have beliefs about the quality of the content. We denote by $\pi_G$ the belief that a given content is good (i.e., of interest or anyway worth accessing) and, conversely, by $\pi_B=1-\pi_G$ the belief that the content is bad. We denote by $\bpi=(\pi_G,\pi_B)$ the corresponding  distribution. Stating $\pi_G=0.75$ means that a user believes that every $4$ similar contents she would get $3$ good ones and $1$ bad one.


\begin{figure}[t]
  \centering
      \includegraphics[width=0.27\textwidth]{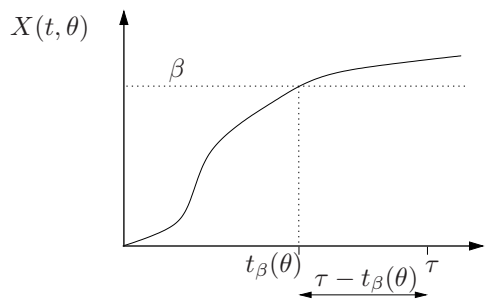}
      \put(-180,100){$X(t,\theta)$}
      \put(-91,10){$t_\beta(\theta)$}\put(-23,10){$\tau$}\put(-65,5){$\tau-t_\beta(\theta)$}  
      \put(-120,84){$\beta$}      
  \caption{The reward or the cost of content $\theta$ for a tagged user is 
           represented by the time during which the content can be accessed, i.e., 
           when viewcount is larger than threshold $\beta$.}\label{fig:util}
\end{figure}
The content access configures as a game where we define {\em players}, {\em strategies} and {\em utilities}. 
{\em Players:} the {\em players} are pull users: based on their belief $\bpi$, they may access the content $\theta$ or not. 

{\em Strategies:} they access $\theta$ when the viewcount is above a certain threshold, i.e., $X(t)\geq \beta \geq 0$.\footnote{We consider the reference case when players select based on the viewcount only for the sake of explanation. We will extend the model to other interesting cases in next sections.} Hence, the {\em strategy} for a certain user is {\fdp the viewcount threshold} $\beta \geq 0$. Of course, all other players also adopt their own strategy with respect to $\theta$ and we denote $\balpha$ the vector of strategies of all remaining users: $\balpha$ is a vector of viewcount thresholds for all other users.

{\em Utilities:} users face either a {\em cost} $C$ or a {\em reward} $R$ for playing strategy $\beta$: the cost and the reward is the fraction of lifetime when the content is in the viewcount range, i.e, when they are willing to access it. The rationale to define this cost/reward is the following. Let a good content be worth one unit reward, and a bad content worth a unit cost. The user may hit several similar contents at random over time. If they are good, the fraction of those actually accessed will be proportional to $1-\frac{t_{\beta}}{\tau}$, where we define $t_{\beta}=\min\{t \,|\,\beta=X(t)\}$, i.e., $t_{\beta}$ is the smallest instant when the threshold is achieved. That also is going to be the long term reward, or the cost, for accessing similar online contents. 
Formally, 
\[
 R(\balpha,\beta,G)=(\tau - t_{\beta}(G))^+, \quad  C(\balpha,\beta,B)=(\tau - t_{\beta}(B))^+
\]
Finally, based on their belief $\bpi$, players expect a utility when playing $\beta$ 
that amounts to 
\begin{eqnarray}
U(\balpha,\beta)= \pi_G R(\balpha,\beta,G) - \pi_B C(\balpha,\beta,B) \nonumber 
\end{eqnarray}

According to the above expression, the cost and the reward are a function of the 
interval when the content is above the threshold, i.e., when the users can 
benefit from it, and depends on the other players strategy. Furthermore, the action 
taken by players depends on their belief on the quality of the content.

In the following we will investigate {\em symmetric equilibria}, i.e., equilibria 
for which all users play $\alpha\geq 0$. We can hence adopt a simplified scalar notation and 
define $t_{\alpha}=\min\{t|\alpha=X(t)\}$.

Let a tagged user playing $\beta$ when all the remaining users use 
$\alpha$: we make the assumption that Wardrop conditions holds. Namely, 
for a large number of users any unilateral deviation of a single user 
does not affect the utilities of other users. I.e., deviations due to 
a single user action are negligible. Wardrop equilibria are much 
easier to compute than the Nash equilibrium; however, Wardrop 
is a good approximation for the latter, as in \cite{Haurie85}.\footnote{A traditional application 
of Wardrop equilibria is road traffic, where users tend to settle to routes minimizing their 
delay: the effect of a route change of an individual driver belonging to a flow is negligible system-wide to the utilities of 
other users.}.

The tagged user expects to gain a certain reward $R(\alpha,\beta,G),$ for a good 
content and expects to suffer a cost $C(\alpha,\beta,B)$ when the content is bad: 
under which conditions $\alpha$ is the best response to itself, namely $\beta^*(\alpha)$?  
We answer to this question in the next sections under different knowledge of the 
viewcount dynamics available to users.

Before we introduce our analysis, we recall that the utility function has the following expression for  $\beta\geq \beta_{\tau,B}$
\begin{eqnarray*}
&&U(\alpha,\beta)=\left\{ 
\begin{array}{ccc}
 0 & \mbox{ if }& \beta \geq  \beta_{\tau,G} \\
 \pi_G (\tau - t_\beta(G)) & \mbox{ if } &     \beta_{\tau,B} \leq \beta \leq \beta_{\tau,G}
  \end{array}
\right.
\end{eqnarray*}
where $ \beta_{\tau,\theta}$  is solution of the following equation 
\begin{equation}
\fdp{t_ {\beta_{\tau}}(\theta) =\tau } 
\end{equation}
We observe that the utility function  $U$ is nonincreasing for $\beta \geq\beta_{\tau,B}$. However the best response  $\beta^*(\alpha)$ can be found only in the interval $[0,\beta_{\tau,B}]$. As a result we restrict our analysis to case when $\beta\leq  \beta_{\tau,B}$ in which the utility function can be expressed as 
$$
U(\alpha,\beta) =  \pi_G (\tau - t_\beta(G))- \pi_B (\tau -t_\beta(B))
$$


\section{Plain Viewcount}\label{sec:gt1}

\fdp{The basic model that we introduce in this section is based on the assumption that pull }
users rely on the number of hits of the contents to judge if it is worth to access it or not, i.e., they judge based on how many users accessed it. Thus, they play based on the dynamics. We hence specialize our analysis to two cases. 
\subsection{Linear case}
First, we examine the case when the process of diffusion of contents is linear. This is the case when the time scale of the content diffusion is very large compared  to the pool of potential users. 
\fdp{A mechanism that that is able generate such a dynamics is the combined effect of an advertisement which 
is broadcasted to a very large pool of viewers, e.g., covering newspapers or other general audience media, 
and people so made aware of the existence of the content who decide to access the content with some random delay thereafter.}

Thus, we let $\Xps(t,\theta)=\Lps t \cdot \One(t)$ where $\One(t)$ is the unitary step function, 
and  $\Xpu(t,\theta)=\Lpu (t-t_{\alpha})\cdot \One(t-t_{\alpha})$\footnote{In a single source
diffusion model, for instance, $X=N(1-\exp(-\lambda t))=N\lambda t + o(t)$}.

Observe that in this case $\Lps=\Lps(\theta)$, whereas $\Lpu$ is independent of $\theta$. In fact, we assume
pull users judge based on viewcount only \cite{Debo2012}. However, we assume that $\Lps(G)\geq \Lps(B)$.

\begin{lemma}\label{lem:noinfo}
In the linear case, under the assumption $\Lpu(G)\geq \Lpu(B)$, it holds
\begin{itemize}
\item[i.] if $\frac{\pi_G}{\Lps(G)} \geq \frac{\pi_B}{\Lps(B)}$, then $\beta^*(\alpha)=0$.  
\item[ii.] if $\frac{\pi_G}{\Lps(G)} \leq \frac{\pi_B}{\Lps(B)}$ but  $\frac{\pi_G}{\Lps(G)+\Lpu} \geq \frac{\pi_B}{\Lps(B)+\Lpu}$
, then $\beta^*(\alpha)=\alpha$
\item[iii.] if $\frac{\pi_G}{\Lps(G)} \leq \frac{\pi_B}{\Lps(B)}$ but  $\frac{\pi_G}{\Lps(G)+\Lpu} < \frac{\pi_B}{\Lps(B)+\Lpu}$
, then $\beta^*(\alpha)=\beta_{\tau,B}$
\end{itemize} 
\end{lemma}
\begin{proof} We need to distinguish two cases, namely $\alpha \geq \beta$ and $\alpha \leq \beta$, determine the best response for each case, and then by comparison choose the best response $\beta^*=\beta^*(\alpha)$. The expression for the utility in the two cases follows. 

If {$\alpha \geq \beta$}, then $X(t,\theta)=\Xps(t,\theta)$ for $0\leq t \leq t_\beta$. 
Thus, we can write simply
\dm{
\[
t_\alpha=\frac{\alpha}{\Lps(\theta)},\quad t_\beta=\frac{\beta}{\Lps(\theta)} 
\]}
and the expression for the utility
\begin{eqnarray}\label{eq:utilcase1}
U(\alpha,\beta)= \tau (\pi_G - \pi_B) - \beta \Big ( \frac{\pi_G}{\Lps(G)} - \frac{\pi_B}{\Lps(B)}\Big ) 
\end{eqnarray} 
\dm{
If {$\alpha \leq \beta$}, then 
$X(t,\theta)=\Xps(t,\theta)$ for $0\leq t \leq t_\alpha$ and $X(t,\theta)=\Xps(t,\theta)+\Xpu(t,\theta)$ for $t_{\alpha}\leq t \leq t_{\beta}$.
In this case,}
\[
t_\alpha=\frac{\alpha}{\Lpu(\theta)},\quad t_\beta=\frac{\beta-\alpha}{\Lps(\theta)+\Lpu}+t_\alpha 
\]
and in turn 
\begin{eqnarray}\label{eq:utilcase2}
U(\alpha,\beta)&&= \tau (\pi_G - \pi_B)- \alpha \Big ( \frac{\pi_G}{\Lps(G)} - \frac{\pi_B}{\Lps(B)} \Big )\nonumber \\
               && - (\beta -\alpha) \Big ( \frac{\pi_G}{\Lps(\theta)+\Lpu} - \frac{\pi_B}{\Lps(\theta)+\Lpu}\Big )
\end{eqnarray}
Now, we can distinguish the three statements in the claim:

{i.} $\frac{\pi_G}{\Lps(G)} \geq \frac{\pi_B}{\Lps(B)}$: in the first case, due to linearity, $\beta=0$ maximizes
the utility; in the second case, we observe that indeed it must hold $\pi_G \geq \pi_B$, and then  
\[
\pi_G\Lps(B)-\pi_B\Lps(G)\geq 0 \geq \Lpu (\pi_B - \pi_G)
\] 
so that  $\frac{\pi_G}{\Lps(G)+\Lpu} \geq \frac{\pi_B}{\Lps(B)+\Lpu}$: in turn the utility function 
is maximized again if $\beta=0$. Hence, it holds $\beta^*(\alpha)=0$.

{ii.} In the first case, it is optimal to maximize $\beta$, which brings $\beta=\alpha$. In the second 
case, in turn it is optimal to minimize $\beta$, so that again $\beta=\alpha$. Hence, $\beta^*(\alpha)=\alpha$.

{iii.} In the first case, the best response is the same as in ii. In the second case, instead, it is optimal to maximize 
$\beta$, so that again $\beta=\beta_{\tau,B}$. However, the last term of (\ref{eq:utilcase2}) is positive and $\beta=\beta_{\tau,B}$ 
maximizes it. Also, by comparison with (\ref{eq:utilcase1}), indeed $\beta^*(\alpha)=\beta_{\tau,B}$ in this case.
\end{proof}

The above results provide a characterization of the possible symmetric Wardrop equilibria of the system.

\begin{thm}\label{ea:linearWardrop}
\begin{itemize}
\item[i.] if $\frac{\pi_G}{\Lps(G)} \geq \frac{\pi_B}{\Lps(B)}$, then $0$ is a symmetric Wardrop equilibrium 
\item[ii.] if $\frac{\pi_G}{\Lps(G)} \leq \frac{\pi_B}{\Lps(B)}$ but  $\frac{\pi_G}{\Lps(G)+\Lpu} \geq \frac{\pi_B}{\Lps(B)+\Lpu}$
,\dm{ then all $0\leq  \beta \leq \beta_{\tau,B}$ are symmetric Wardrop equilibria }
\item[iii.] if $\frac{\pi_G}{\Lps(G)} \leq \frac{\pi_B}{\Lps(B)}$ but  $\frac{\pi_G}{\Lps(G)+\Lpu} < \frac{\pi_B}{\Lps(B)+\Lpu}$
, \dm{ then $\beta_{\tau,B}$ is a symmetric Wardrop equilibrium }
\end{itemize} 
\end{thm}

It is possible to interpret the above result as follows: $\frac{\pi_G}{\Lps(G)}$ represents the time pace at which push users are believed to access a good content. Similarly $\frac{\pi_B}{\Lps(B)}$ represents the time pace at which push users are believed to access a bad content. Thus, condition i. suggests that it is always convenient to anticipate the access to the content. In case ii., the situation is dictated by the uptake of pull users, because they increase the viewcount thus reinforcing the believed viewcount pace of a good content against that of a bad content. Finally, in case iii. there is no incentive in accessing the content.

\subsection{Exponential case:  fixed time horizon}

Let us consider the content dissemination process operated by a content provider using a finite set of potential target users. After the content is posted by the provider directly to users, it will be transmitted to more and more users by using some preferential channels. In this case, we need to model the push dynamics accounting for the size $N$ of the pool of push users, i.e., we assume that the content provider disseminates the content according to
\[ 
{\dot X_{ps}}(t,\theta)= \Lps(\theta)(N-\Xps(t,\theta)), 
\] 
so that
\begin{equation}\label{expo1}
\Xps(t,\theta)=N(1 - e^{-\Lps(\theta) t} )\mbox{ for }  t\geq 0 
\end{equation}

We reported in Fig.~\ref{fig:expo_utility} the shape of the utility function under the exponential case \fdp{for a fixed time horizon.}
As it can be observed in case a), for smaller values of $\alpha$, i.e, $\alpha=400$ a low value of the belief $\pi_G$ causes the access to be delayed till time $\tau$, whereas for increasing  values of $\pi_G$ we observe first a local maximum at $\alpha$ ($\pi_g=0.75$), and finally the strategy $\beta=0$ takes over corresponding to very large values of $\pi_G$. Indeed, such a behavior of the utility function resembles -- for a fixed $N$ -- what we observed in the linear case. However, at a closer look, namely in  Fig.~\ref{fig:expo_utility}c) we understand that the situation is more elaborate: in particular, we know that number of push users $N$ impacts the speed at which the viewcount increases. As such, a small $N$ \fdp{does not} permit to pass the threshold $\alpha$, whereas a very large one \fdp{incentivizes} early access: recall that $\beta_{\max}:=\beta_{\tau,B}$ means access at time $t=0$. In between, the presence of a maximum predicts, as in the linear case, the existence of best responses that lie in the interior of $[0,\beta_{\max}]$. This intuitive numerical insight is confirmed by the theoretical results that we detail in the following.

\begin{figure*}[t]
  \centering
  \subfigure[Case $\alpha=400$]{\includegraphics[width=0.30\textwidth]{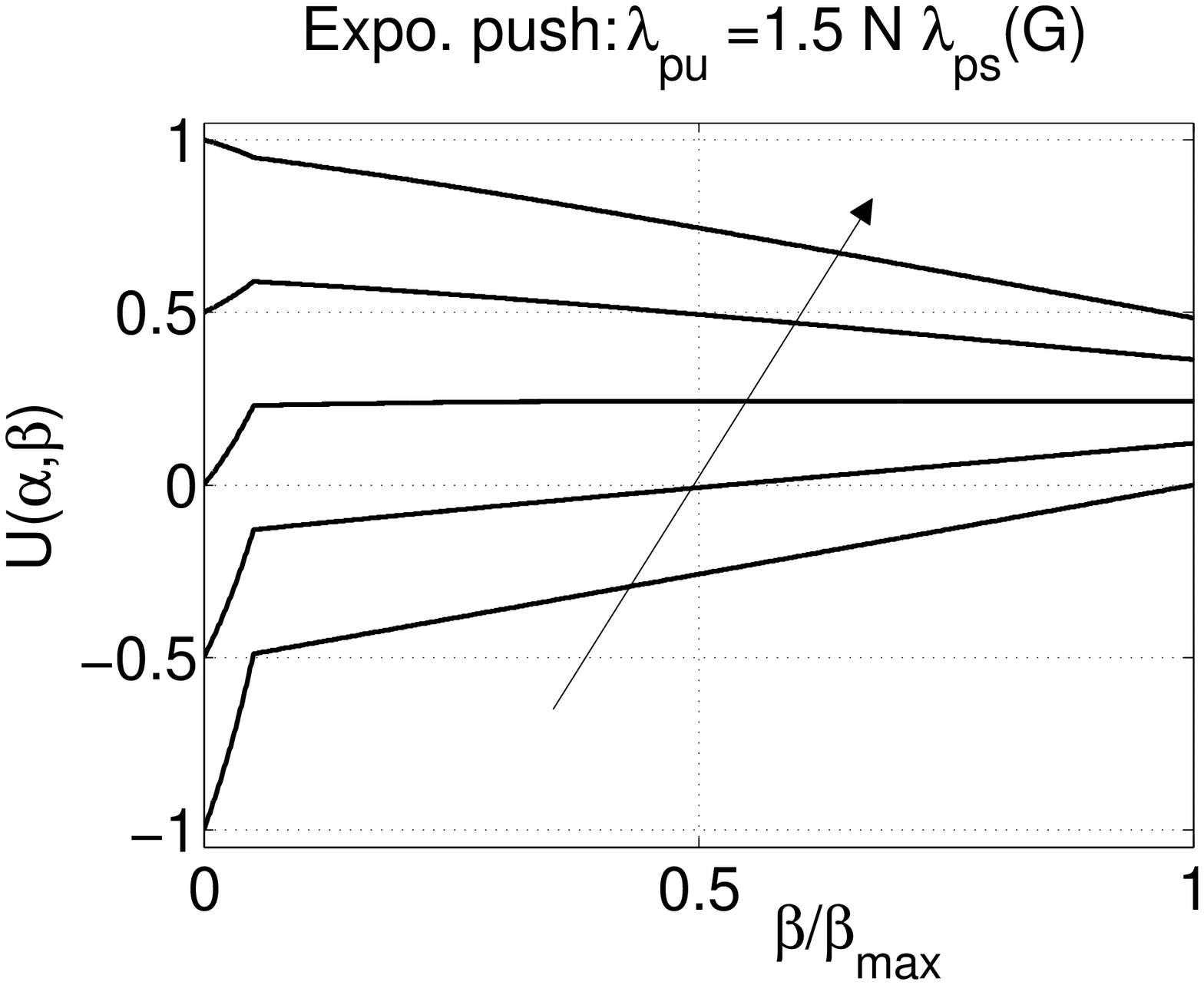}\put(-110,35){\scriptsize $\pi_G=0,0.25,0.5,0.75, 1$}}
\subfigure[Case $\alpha=700$]{\includegraphics[width=0.30\textwidth]{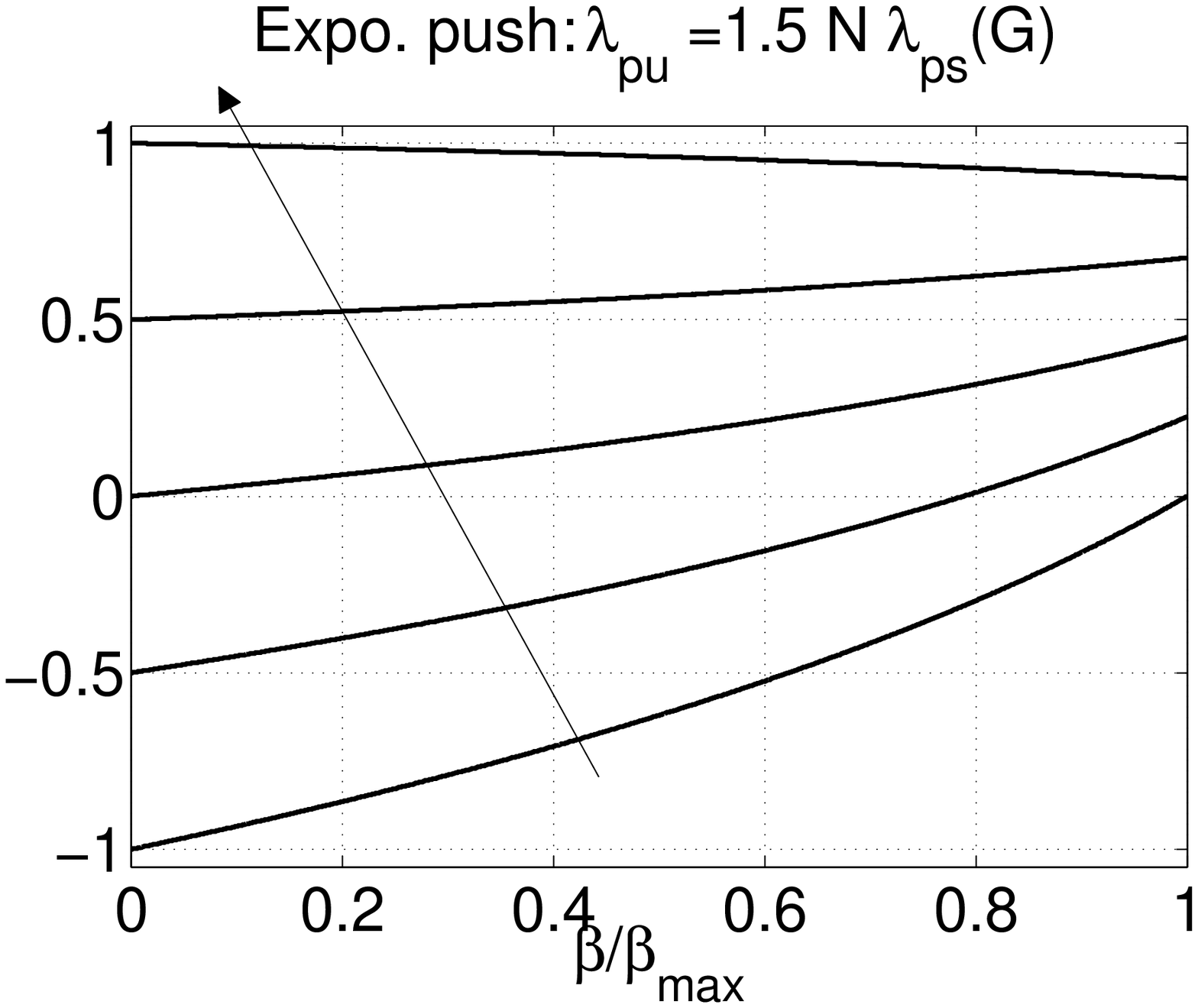}\put(-90,28){\scriptsize $\pi_G=0,0.25,0.5,0.75, 1$}}
\subfigure[$\alpha=700$, increasing $N$]{\includegraphics[width=0.30\textwidth]{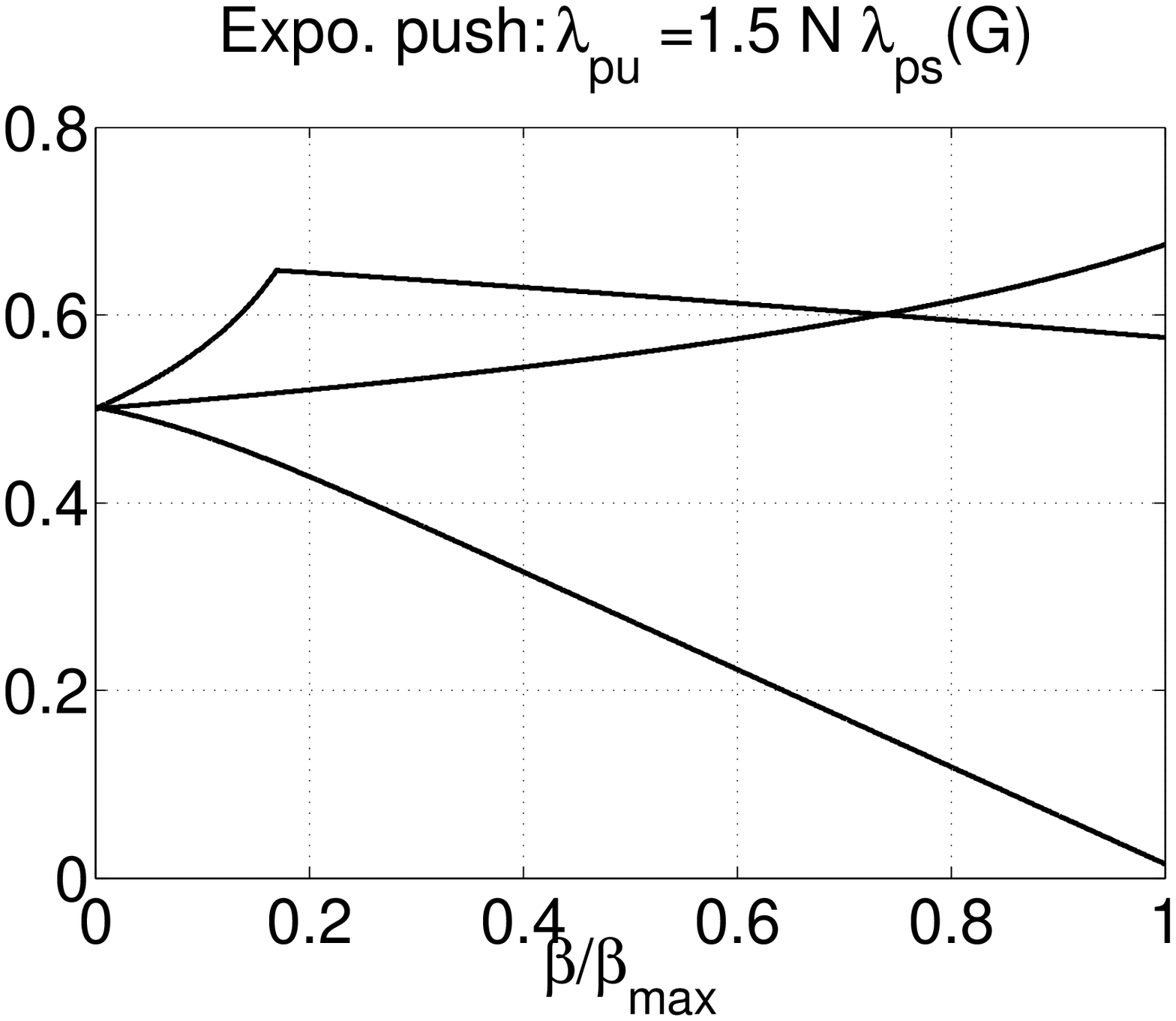}\put(-50,45){\scriptsize $N=50000$}\put(-120,105){\scriptsize $N=1000$}\put(-97,80){\scriptsize $N=700$}}
\caption{The utility function for $N=1000$, for $\tau=10$ days, $\Lps(G)=10^{-1}$ views/day, $\Lps(B)=\Lps(G)/10$. a) $\alpha=400$ views, b) $\alpha=700$ views. Increasing values of the belief $\pi_G$ determine different shapes for the utility function. c) Increasing values of $N=700,1000,50000$ for $\alpha=700$. All graphs for $\Lpu =1.5 N \Lps(G)$.}\label{fig:expo_utility}
\end{figure*}

We distinguish two cases, namely $\alpha<\beta$ and $\beta \leq \alpha$.  

If {$\beta \leq \alpha$},  we have 
\[
t_{\beta}(\theta)=-\frac 1{\Lps(\theta)}\log \Big ( 1- \frac \beta N  \Big ), \;\;
t_{\alpha}(\theta)=-\frac 1{\Lps(\theta)}\log \Big ( 1- \frac \alpha N  \Big )
\]
Hence the utility becomes 
\[
U(\alpha,\beta)=(\pi_G-\pi_B)\tau + \log \Big ( 1- \frac \beta N  \Big ) \Big ( \frac{\pi_G}{\Lps(G)} - \frac{\pi_B}{\Lps(B)}\Big )
\]
\dm{Let $\beta_1^*(\alpha)$ (resp. $\beta_2^*(\alpha)$) be the best response to $\alpha$ in $[0,\alpha]$ (resp. $[\alpha,\beta_{max}]$)}
\begin{lemma}
In the exponential case, under the assumption $\lambda_{ps}(G)>\lambda_{ps}(B)$, it holds for $\beta \leq \alpha$
\begin{itemize}
\item If $\frac{\pi_G}{\pi_B} < \frac{\lambda_{ps}(G)}{\lambda_{ps}(B)} $ then $\beta^*_1 (\alpha)=\alpha$
\item  If $\frac{\pi_G}{\pi_B} > \frac{\lambda_{ps}(G)}{\lambda_{ps}(B)} $ then $\beta^*_1 (\alpha)=0$ 
\item  If $\frac{\pi_G}{\pi_B}= \frac{\lambda_{ps}(G)}{\lambda_{ps}(B) }$ then for every $\beta^*_1 \in [0,\alpha]$ is optimal
\end{itemize}
\end{lemma}
\begin{proof} \fdp{The proof is similar to the one developed in the linear case for $\beta \leq \alpha$.}
\end{proof}

Now, we  study the second case: {$\alpha \leq \beta$}.  If $t_\alpha \leq t \leq t_{\beta}$, 
\begin{equation}\label{expo2}
X(t,\theta)= N(1 - \exp(-\Lps(\theta) t) + \Lpu (t-t_{\alpha})
\end{equation} 
for which we obtain 
\begin{eqnarray}
t_{\beta}&&=\Lps(\theta) \Big(W\Big( \frac{\Lps(\theta)}\Lpu N  \frac{e^{\frac{\Lps(\theta)}\Lpu N(1-\frac \beta N)}}{\big (1-\frac \alpha N \big )}\Big ) \nonumber\\ 
                    && -  \log \Big ( \frac{e^{\frac{\Lps(\theta)}\Lpu N\big (1-\frac \beta N \big )}}{\big ( 1- \frac{\alpha}{N} \big )} \Big )\Big)
\end{eqnarray}
where $W(\cdot)$ is the Lambert function \cite{CorlessLambert}.  We can obtain the derivative of the above expression by letting $\xi(\beta)= \frac{e^{\zeta(\theta)(1-\frac \beta N)}}{\big (1-\frac \alpha N \big )}$
and $\zeta(\theta) = \frac{\Lps(\theta)}\Lpu N$
\begin{eqnarray}
&&\frac d{d\beta} t_{\beta} =\frac 1 {\Lps(\theta)} \frac d{d\beta} W(\zeta(\theta)\xi(\beta, \theta)) - \log(\xi(\beta))\nonumber \\ 
                                     &&=\frac 1 {\Lpu} \cdot \frac 1{1 + W(\zeta(\theta)\xi(\beta, \theta))} \nonumber
\end{eqnarray}
\fdp{After some cumbersome algebra, we derive}
\begin{lemma}
\fdp{In the exponential case, under the assumptions $\lambda_{ps}(G)>\lambda_{ps}(B)$ and  $\lambda_{ps}(G)N \leq \lambda_{pu}$, for $\alpha \leq \beta$ it holds}
 \dm{\begin{itemize}
\item   If ${\pi_G}\leq \pi_B $ then $\beta^*_2(\alpha)=\beta_{\tau,B}$ 
\item If $\frac{1+W(\zeta(G)\xi(\alpha, G))}{1+W(\zeta(B)\xi(\alpha,B))}\geq \frac{\pi_G}{\pi_B}$ for all $\beta \in [\alpha, \beta_{\tau,B}]$  then $º\beta^*_2(\alpha)= \alpha$
\item If $\frac{1+W(\zeta(G)\xi(\beta_\tau, G))}{1+W(\zeta(B)\xi(\beta_tau, B))}\leq \frac{\pi_G}{\pi_B}$  for all $\beta \in [\alpha, \beta_\tau(B)]$  then $\beta^*_2(\alpha)=\beta_{\tau,B}$
\item otherwise $\beta^*_2(\alpha)$ is the solution of the following equation
$$\frac{1+W(\zeta(G)\xi(\beta^*_2(\alpha), G))}{1+W(\zeta(B)\xi(\beta^*_2(\alpha), B))}=\frac{\pi_G}{\pi_B}$$
\end{itemize}}
\end{lemma}
 \begin{proof} 
 The derivative of the utility function $U$ is  
\begin{equation}
\label{derivative}
U'(\alpha, \beta) =\frac{1}{\lambda_{pu}}\Big(\frac{\pi_B}{W(\zeta(G)\xi(\beta, B))}-\frac{\pi_B}{W(\zeta(G)\xi(\beta, G))}\Big)
\end{equation}
Since  $\xi(\beta, G)>\xi(\beta, B)$ and $\zeta(G)>\zeta(B)$ then it is easy to check under condition ${\pi_G}\leq \pi_B $ that $U'(\alpha, \beta) >0$. Hence the utility function attains a unique maximum  at   $\beta_{\tau,B}$. 

In order to complete the proof, it is sufficient to show that the function $U$  is  either  non-increasing, or there is some $\bar \beta$ such that  $U$  is non-decreasing for $\beta<\bar \beta$  and non-increasing for $\beta>\bar \beta$. 

Assume that there exists a $\bar \beta$ such that  $U'(\alpha, \bar \beta)\leq 0$. From (\ref{derivative}), it is sufficient to show that 
$$
U'(\alpha,  \beta)\leq 0  \;\,\;\mbox{ for all } \beta >\bar \beta
$$


We can show the above propriety by letting $\bar W(\beta) = \frac{1+W(\zeta(G)\xi(\beta, G))}{1+W(\zeta(B)\xi(\beta,B))}$ and it turns out that 
\begin{eqnarray*}
&&\frac{\partial\bar W(\beta)}{\partial \beta}= \frac{1}{(1+W(\zeta(B)\xi(\beta,B)))^2}\\
&&\Big(\frac {\zeta(B) W(\zeta(B)\xi(\beta, B))(1+W(\zeta(G)\xi(\beta, G)))}{1+W(\zeta(B)\xi(\beta, B))}\\
&&-\frac {\zeta(G) W(\zeta(G)\xi(\beta, G))(1+W(\zeta(B)\xi(\beta, B)))}{1+W(\zeta(G)\xi(\beta, G))}\Big)
\end{eqnarray*}
To show  $\frac{\partial\bar W(\beta)}{\partial \beta}\leq 0$, we impose the inequality 
\begin{equation}
\label{compa}
\frac {\zeta(B) W(\zeta(B)\xi(\beta, B))}{(1+W(\zeta(B)\xi(\beta, B)))^2} \leq  \frac {\zeta(G) W(\zeta(G)\xi(\beta, B))}{(1+W(\zeta(G)\xi(\beta, G)))^2}
\end{equation}
We can obtain the above inequality under assumption $\lambda_{ps}(G)N \leq \lambda_{pu}$ by letting
$$ f(y)=\frac {y W(y \frac{e^{y(1-\frac{\beta}{N})}}{(1-\frac{\alpha}{N}})}{(1+W(y \frac{e^{y(1-\frac{\beta}{N})}}{1-\frac{\alpha}{N}}))^2}$$
Hence the derivative of $f$ 
can be expressed as 
\begin{equation}
\frac{\partial f}{\partial y} = w(\bar y)\frac{w^2(\bar y)+ w(\bar y) (1-y (1-\frac{\beta}{N}))+2+y (1-\frac{\beta}{N})}{(1+w(\bar y)^2}
\end{equation}
where $\bar y=y \frac{e^{y(1-\frac{\beta}{N})}}{(1-\frac{\alpha}{N})}$. In fact it can be showed that $\dot f$ is positive for $y(1-\frac{\beta}{N})\leq 1$ i.e.,  $\lambda_{ps}(G)N \leq \lambda_{pu}$. 
\end{proof}

\dm{Overall, the above cases are summarized in the following theorem
\begin{thm}\label{thm:expo}
Let  $\lambda_{ps}(G)>\lambda_{ps}(B)$ and  $\lambda_{ps}(G)N \leq \lambda_{pu}$, then in the exponential case  
\begin{itemize}
\item[i)] If ${\pi_G}\leq \pi_B$ then $\beta_{\tau,B}$ is a symmetric Wardrop equilibrium
\item[ii)] If ${\pi_G}> \pi_B $ then the following cases hold
\begin{itemize}
\item[a)]  If  $\frac{\pi_G}{\pi_B} < \frac{\lambda_{ps}(G)}{\lambda_{ps}(B)} $ and  $\frac{1+W(\zeta(G)\xi(\alpha, G))}{W(\zeta(B)\xi(\alpha,B))}\geq \frac{\pi_G}{\pi_B}$ for all $\beta \in [\alpha,\beta_{\tau,B}]$  then all $0< \beta \leq \beta_{\tau,B}$ are symmetric Wardrop equilibria 
\item[b)]  If $\frac{\pi_G}{\pi_B} < \frac{\lambda_{ps}(G)}{\lambda_{ps}(B)} $ and $\frac{1+W(\zeta(G)\xi(\beta_\tau, G))}{1+W(\zeta(B)\xi(\beta_\tau, B))}\leq \frac{\pi_G}{\pi_B}$  for all $\beta \in [\alpha, \beta_{\tau,B}]$  then $\beta_{\tau,B}$  is a symmetric Wardrop equilibrium
\item[c)]  If $\frac{\pi_G}{\pi_B} < \frac{\lambda_{ps}(G)}{\lambda_{ps}(B)} $ and there exists a  $\bar \beta$ is the solution of the following equation
$$\frac{1+W(\zeta(G)\xi(\bar\beta, G))}{1+W(\zeta(B)\xi(\bar \beta, B))}=\frac{\pi_G}{\pi_B}$$
then $\bar \beta$  is a symmetric Wardrop equilibrium
\end{itemize} 
\item[iii)] If $\frac{\pi_G}{\pi_B} >  \frac{\lambda_{ps}(G)}{\lambda_{ps}(B)}$, then the following cases hold
\begin{itemize}
\item[a)] if $\frac{1+W(\zeta(G)\xi(\alpha, G))}{1+W(\zeta(B)\xi(\alpha,B))}\geq \frac{\pi_G}{\pi_B}$ for all $\beta \in [\alpha,\beta_{\tau,B}]$  then $0$ is a symmetric Wardrop equilibrium
 \item[b)] if $\frac{1+W(\zeta(G)\xi(\alpha, G))}{1+W(\zeta(B)\xi(\alpha,B))}\leq \frac{\pi_G}{\pi_B}$ for all $\beta \in [\alpha, \beta_{\tau,B}]$, then there exists a symmetric Wardrop equilibrium which is given by 
\begin{eqnarray}
\left\{
\begin{array}{cc}
0 & \mbox{ if } \tau \pi_B < \pi_G t_{\beta_{\tau,B}}(G)\\
\beta_{\tau,B}& \mbox{ if } \tau \pi_B > \pi_G t_{\beta_{\tau,B}}(G)\\
\beta^* \in \{0, \beta_{\tau,B}\} & \mbox{ if } \tau \pi_B = \pi_G t_{\beta_{\tau,B}}(G)
\end{array}
\right.
\end{eqnarray}
\end{itemize}
\end{itemize}
\end{thm}}

\fdp{Theorem.~\ref{thm:expo} displays a structure of the best response  that is similar to the result obtained for the linear case, but 
we should highlight some differences. First, the additional request $\lambda_{ps}(G)N \leq \lambda_{pu}$
is excluding the case when the effect of the pull mechanism is negligible compared to push mechanism. 
This means that we are restricting to the case when the aggregated maximum rate at which the viewcount can increase 
due to the push mechanism is smaller than the increase that is generated once the viewcount is above threshold for pull users. 
Indeed, this is the interesting case when the content provider's aim is to attract a large basin of pull users using a 
target limited audience of push users.} 

\fdp{Second, we observe that the term  $\frac{\pi_\theta}{\lambda_{ps}(\theta)+ \Lpu}$ that was present in the linear case is 
now replaced by a term involving the Lambert function $W(\cdot)$ \cite{CorlessLambert}: this is due to the combined effect of the exponential growth and the linear growth above the threshold, accounting for the saturation of the basin of push users. In the case 
when $N$ is very large or $\lambda_{ps}$ is very small, the term collapses to the condition expressed in the linear case. }

\section{Variable time horizon}\label{sec:gt2}

In this section, we are interested in the case where the time horizon \fdp{during which the content is accessed by pull users is not fixed. But, it is determined by the popularity of the content and by the  quality perceived by users. In particular, when the popularity of a content is subject to saturation, we can model a vanishing $\dot X$ to encode the condition when a content which is present online for a long time becomes stale. Conversely, fresh uptaking contents will experience  large values of $\dot X$ and will be preferred.} \fdp{This case fits well specific types of contents such as news or pop songs, for which the {\em trend} of the viewcount increase may be the main trigger for the users' interest in some content. Pull users still adopt a threshold strategy and browse the content if}
\begin{equation}\label{eq:threxpo}
\dot X(t,\theta) \geq \gamma_{th}
\end{equation}
Let us consider the exponential push case introduced in the previous section. Condition \eqref{eq:threxpo} determines a variable horizon to access content $\theta$:
$$
\tau(\alpha, \theta) =\dot X^{-1}( \gamma_{th})
$$
Because the time horizon $\tau=\infty$ for $ \gamma_{th}\leq \lambda_{pu}$, we restrict our analysis to the case when $\gamma_{th}>\lambda_{pu}$.

 \fdp{Again, we are interested to compute the utility function for a tagged user given a certain common threshold strategy $\alpha$ played by other 
 users; the objective is to compute the best response $\beta$ for the tagged user as done before}. Let   $X_{th} (\theta)= N-\frac{\gamma_{th}}{\lambda_{ps}(\theta)}$, $\tau_0(\theta) =     \frac{1}{\lambda_{ps}(\theta)} \log\big(\frac{\lambda_{ps}(\theta) N}{\gamma_{th}}\Big)$ and $\tau_1( \theta) =    \frac{1}{\lambda_{ps}(\theta)} \log\big(\frac{\lambda_{ps}(\theta) N}{\gamma_{th}-\lambda_{pu}}\Big)$. 
 
 \fdp{Observe that the interval of time when pull users will access the content becomes now $[\tau_0(\theta)\, \tau_1(\theta)]$: the duration of 
 such interval corresponds to the useful lifetime of the content as dictated by the interest of the users based on \eqref{eq:threxpo} and by the content type.} 
 
\fdp{We distinguish again two intervals, namely $0\leq \beta \leq \alpha$ and $\alpha\leq \beta \leq \tau$, and denote $\beta_1^*$ and $\beta_2^*$ the best response in those intervals, respectively. However, we need to account also for \eqref{eq:threxpo} and to detail the utility accordingly.}

It follows that if $\beta\geq \alpha$, then 
$$
U(\alpha,\beta)= \pi_G \Big(\tau_1(G) -t_\beta (G)\Big)^+ - \pi_B \Big(\tau_1(B) -t_\beta (B)\Big)^+
$$
 
If $\alpha > X_{th} (G)$ and $\beta\leq \alpha$
\begin{eqnarray*}
U(\alpha,\beta)&=& \pi_G \Big(\tau_0(G) -t_\beta (G)\Big)^+ + \pi_G  \Big(\tau_1(G)-t_\alpha(G)\Big)^+ \\
&&- \pi_B \Big(\tau_0(B) -t_\beta (B)\Big)^+ - \pi_B  \Big(\tau_1(B)-t_\alpha(B)\Big)^+
\end{eqnarray*}
If $X_{th}(B)\leq \alpha \leq X_{th} (G)$ and $\beta\leq \alpha$, then
\begin{eqnarray*}
U(\alpha,\beta)&=& \pi_G \Big(\tau_1(G) -t_\beta (G)\Big) \\
&&- \pi_B \Big(\tau_0(B) -t_\beta (B)\Big)^+ - \pi_B  \Big(\tau_1(B)-t_\alpha(B)\Big)^+
\end{eqnarray*}
 If $X_{th}(B)\geq \alpha$ and $\beta\leq \alpha$, then
\begin{eqnarray*}
U(\alpha,\beta)&=& \pi_G \Big(\tau_1(G) -t_\beta (G)\Big)  - \pi_B \Big(\tau_1(B)-t_\alpha(B)\Big)
\end{eqnarray*}
\begin{figure*}[t]
  \centering
  \subfigure{\includegraphics[width=0.30\textwidth]{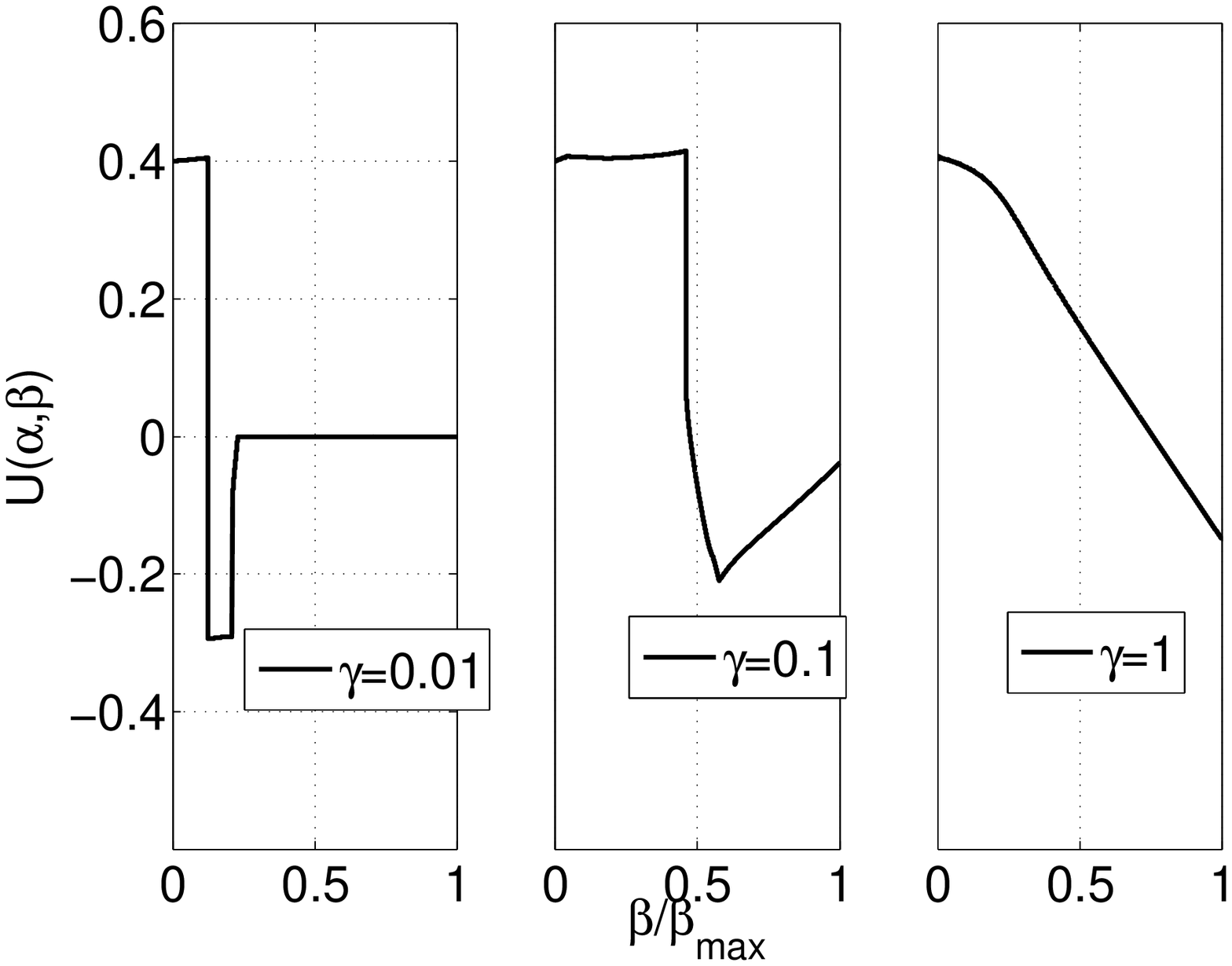}}
\subfigure{\includegraphics[width=0.30\textwidth]{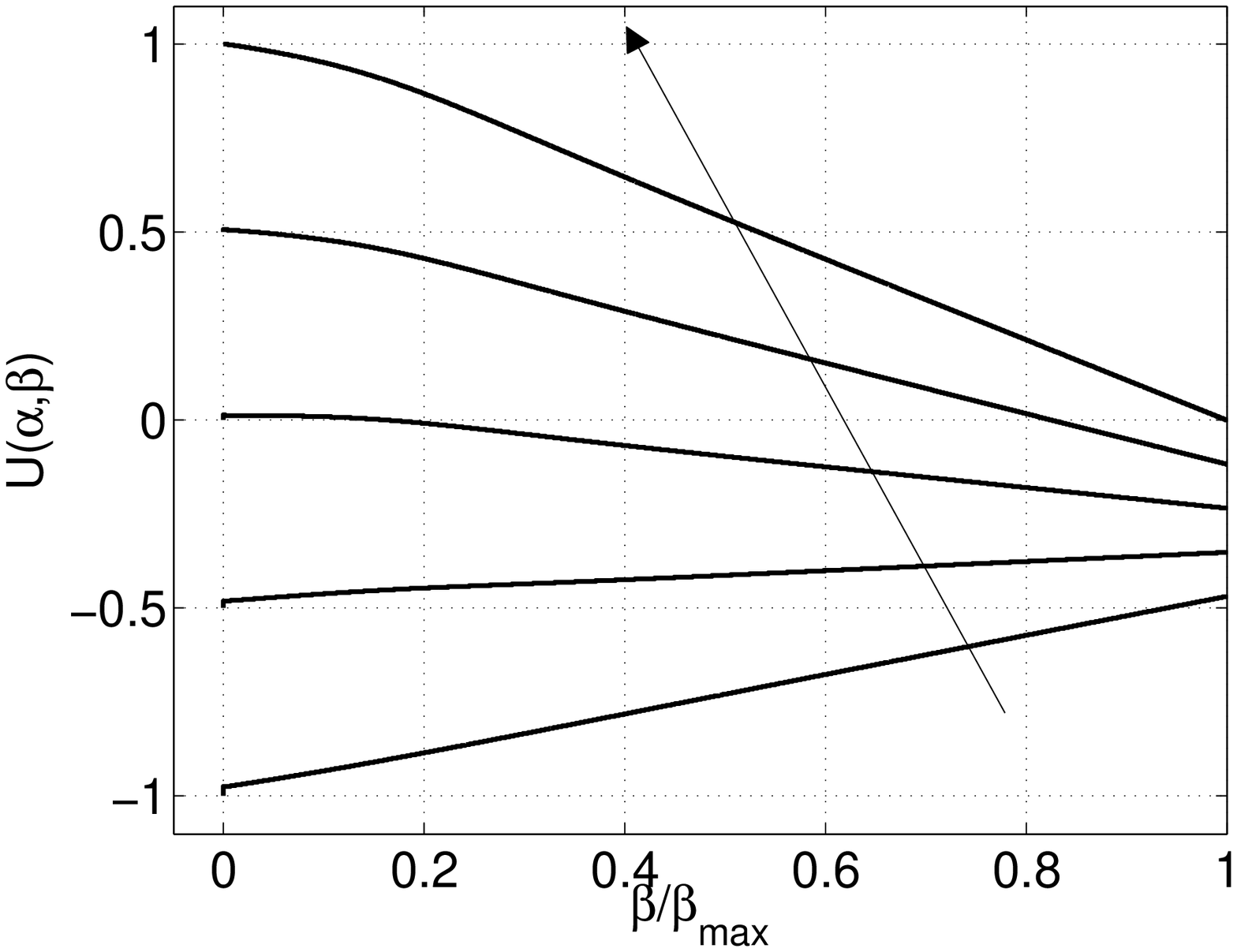}\put(-85,22){\scriptsize $\pi_G=0,0.25,0.5,0.75, 1$}}
\subfigure{\includegraphics[width=0.31\textwidth]{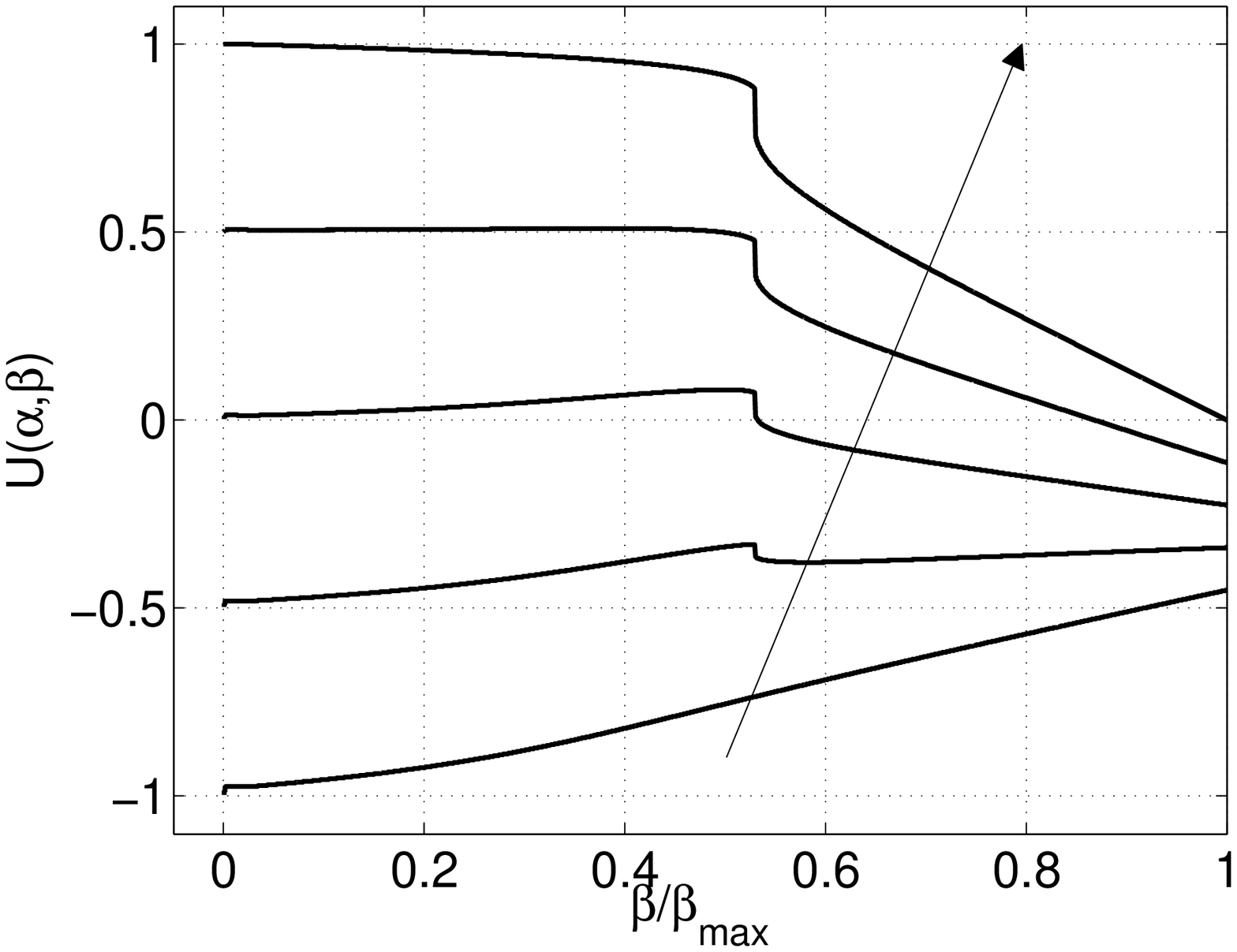}\put(-90,19){\scriptsize $\pi_G=0,0.25,0.5,0.75, 1$}
 \put(-155,110){(c)}\put(-315,110){(b)}\put(-480,110){(a)}}
\caption{The utility function for $N=1000$, for $\tau=10$ days, $\Lps(G)=10^{-1}$ views/day, $\Lps(B)=\Lps(G)/10$. a) Detail of the discontinuities of $U(\alpha,\beta)$ for $\gamma=0.01,0.1,1$, where $\alpha=0.18$ b) Extremal type of best response for $\alpha=0.029$, $\gamma=1.5$ and under increasing values of the belief $\pi_G$. c) Same as b) but for $\gamma=0.3$. Discontinuity in $\alpha$ corresponds to local maxima for $\pi_G=0.25,0.50$.}\label{fig:expo_xxdot}
\end{figure*}

\dm{With a similar analysis as that employed in the proof of  Thm.\ref{thm:exp_thr}, we can write:}
\begin{thm}\label{thm:exp_thr}
In the exponential case, under the assumptions $\lambda_{ps}(G)>\lambda_{ps}(B)$ and  $\lambda_{ps}(G)N \leq \lambda_{pu}$, it holds
\begin{itemize}
\item \fdp{If ${\pi_G}\leq \pi_B $ then $\beta$   is a symmetric Wardrop equilibrium where $\beta=\beta_\tau(B)$ is solution of $t_{\beta}(B)=\tau $}
\item \dm{ If ${\pi_G}> \pi_B $,  $\frac{\pi_G}{\pi_B} < \frac{\lambda_{ps}(G)}{\lambda_{ps}(B)} $ and  $\frac{1+W(\zeta(G)\xi(\alpha, G))}{1+W(\zeta(B)\xi(\alpha,B))}\geq \frac{\pi_G}{\pi_B}$ for all $\beta \in [\alpha, \beta_{\tau_{0,B}}]$  then all values in the interval  $[0, \tilde\beta] $ are symmetric Wardrop equilibria.}

\item  \dm{ If ${\pi_G}> \pi_B $,   $\frac{\pi_G}{\pi_B} < \frac{\lambda_{ps}(G)}{\lambda_{ps}(B)} $ and $\frac{1+W(\zeta(G)\xi(\beta_\tau, G))}{1+W(\zeta(B)\xi(\beta_\tau, B))}\leq \frac{\pi_G}{\pi_B}$  for all $\beta \in [\alpha,\beta_{\tau_{0,B}}(B)]$  then $\tilde \beta$ is a  symmetric Wardrop equilibrium.}

\item \dm{ If ${\pi_G}> \pi_B $,  $\frac{\pi_G}{\pi_B} < \frac{\lambda_{ps}(G)}{\lambda_{ps}(B)} $ and there exists $ \beta_s$ solution of the following equation
$$\frac{1+W(\zeta(G)\xi(\beta_s, G))}{1+W(\zeta(B)\xi( \beta_s, B))}=\frac{\pi_G}{\pi_B}$$
then $\beta_s$ is a  symmetric Wardrop equilibria.}
\dm{\item  If $\frac{\pi_G}{\pi_B} >  \frac{\lambda_{ps}(G)}{\lambda_{ps}(B)} $ and  and  $\frac{1+W(\zeta(G)\xi(\alpha, G))}{1+W(\zeta(B)\xi(\alpha,B))}\geq \frac{\pi_G}{\pi_B}$ for all $\beta \in [\alpha,\beta_{\tau_{0,B}}]$  then $0$  is a  symmetric Wardrop equilibrium.
\item  If $\frac{\pi_G}{\pi_B} >  \frac{\lambda_{ps}(G)}{\lambda_{ps}(B)} $ and  and  $\frac{1+W(\zeta(G)\xi(\alpha, G))}{1+W(\zeta(B)\xi(\alpha,B))}\leq \frac{\pi_G}{\pi_B}$ for all $\beta \in [\alpha, \beta_{\tau_{0,B}}]$  then there exists a  symmetric Wardrop equilibrium which is given by}
\begin{eqnarray}
\left\{
\begin{array}{cc}
0 & \mbox{ if } \tau \pi_B < \pi_G t_{\beta_{\tau_{0,B}}}\\
\beta_\tau(B)  & \mbox{ if } \tau \pi_B > \pi_G t_{\beta_{\tau_{0,B}}}\\
\beta^* \in \{0, \beta_{\tau_0}(B)\} & \mbox{ if } \tau \pi_B = \pi_G t_{\beta_{\tau{0,B}}}
\end{array}
\right.
\end{eqnarray}
\end{itemize} 
\end{thm}

\fdp{The overall result in Thm.\ref{thm:exp_thr} shows a structure that is close to that obtained in Thm.~\ref{thm:expo}. We can conclude 
that the presence of a selective preference expressed in terms of the viewcount trend does not affect the structure of the Wardrop 
equilibria. In fact, they are of the kind determined before in the case of a fixed length interval: either extremal ones or a continuum of such restpoints. It is interesting to notice that this is following irrespective of the fact that the utility function is linear as a function of the "viewing time", i.e., the time that is useful for the viewers, but, pull users' preferences depend on a non-linear function of the threshold type.}

\section{Combined effect of Trend and Viewcount}\label{sec:gt3}

In general, contents that are present online since a long time display different 
popularity than contents which last only a short time \cite{SzaboPop}. As we noticed in the previous 
section, when popularity saturation occurs, $\dot X$ vanishes for large $t$. \fdp{If users choose among contents with 
different trend and different viewcount, they would naturally choose a content with large viewcount and large increasing trend.  
To this respect $y(t)=\dot X(t) X(t)$ encodes the condition when the pull user still values the viewcount, but, she favors a large 
increasing trend given two contents with the same viewcount.}

Symmetric equilibria can be determined when in the system all users adopt a strategy 
\[
\alpha := y(t_\alpha), \quad 0 \leq t_\alpha \leq \tau 
\]
and again we determine the best response for a user deviating using $\beta := y(t_\beta)$ as a reply, where  $0 \leq t_\beta \leq \tau$.

It is easy to see that in the linear case, the model developed in the previous section 
applies as long as one replaces the dynamics with the one below 
\[
\Xps(t,\theta)=\Lps^2(\theta)t+\Lps(\theta), \quad \Xpu=\Lpu^2(t-t_{\alpha})\cdot \One(t-t_{\alpha})
\]
so that all the results can be specialized accordingly replacing $\Lps$ and $\Lpu$ with $\Lps^2$ and $\Lpu^2$
wherever they appear. The intuition is that when the regime of content diffusion is linear, i.e., 
when a large number of push users exists, the trend of popularity has the only effect to reinforce 
the inequality $\Lps(B)\not =\Lps(G)$. We then move to a more interesting case.

\subsection{Exponential push case}

In the exponential case, the dynamics again is the same captured by (\ref{expo1}), (\ref{expo2}). We 
can specialize the analysis to the two cases as done before. If {$\alpha \geq \beta$}, $y(t_\beta)=\beta$ 
implies that
\[
\beta=\Lps(\theta) N^2(1-e^{-\Lps(\theta) t_\beta})e^{-\Lps(\theta) t_\beta}
\]
where the solution is such that $t_\beta=-\frac1{\lambda(\theta)} f(\beta,\theta)$ where we let $f(\beta,\theta):=\log \Big ( \frac 12 \Big ( 1 + \sqrt{1 - \frac{4\beta}{\Lps(\theta)N^2}} \Big ) \Big )$. 
\[
U(\alpha,\beta)=(\pi_G-\pi_B)\tau + \Big ( \frac{\pi_G}{\Lps(G)}f(\beta,G) - \frac{\pi_B}{\Lps(B)}f(\beta,B) \Big )
\]
After observing that $f(0,\theta)=0$ and $f(\beta,G)\leq f(\beta,B) \leq 0$, again we obtain two extremal cases: 
when $\frac{\pi_G}{\Lps(G)} \geq \frac{\pi_B}{\Lps(B)}$ then $U(\alpha,\beta)-(\pi_G-\pi_B)\tau\leq 0$ so that 
$\beta=0$ maximizes the utility. In the opposite case, namely, $\frac{\pi_G}{\Lps(G)} \leq \frac{\pi_B}{\Lps(B)}$, $U(\alpha,\beta)-(\pi_G-\pi_B)\tau\geq 0$, so that $\beta=\alpha$ does. 

If $\alpha \leq \beta$, the condition for 
\[
\beta=X(t) \Big ( \Lps(\theta) N e^{-\Lps(\theta)} + \Lpu \Big ) 
\]
\begin{eqnarray}\label{eq:utexpxdotx}
\mbox{gives:} \quad t_\beta 
      =t_\alpha(\theta) - \frac N{\Lpu} \Big [1- \frac{W(f(\beta,B)\xi(B)e^{-\xi(B)})}{\xi(B)e^{-2\xi(B)}} \Big ]\nonumber
\end{eqnarray}
where we used the definition of $We^W=x$ and we stressed the dependence of $t_\alpha$ on $\theta$. It is important to notice that in this case, $t_\beta$ is not continuous, so that in correspondence of $t_\alpha(G)$ and $t_\alpha(B)$ the utility function has possibly two discontinuities. We reported in Fig.~\ref{fig:expo_xxdot}(a) the shape of the utility function for increasing values of $\gamma=0.01,0.1,1$ $\Lpu =\gamma N \Lps(G)$. For larger values of $\gamma$ the effect of discontinuities becomes negligible with respect to the shape of the utility function (indeed we are looking for the best response, i.e., the maximum of $U(\alpha,\beta)$). 

In particular, we observe in Fig.~\ref{fig:expo_xxdot}(b) that for the choice of parameters there, i.e., $\gamma=1.5$, the shape of the utility function leads again to the customary extremal type of best response that we observed in the linear case. That is, access at time $t=0$, i.e., $\beta=\beta_{\max}$ for large $\pi_G$ and access at time $t_\beta=\tau$, i.e., $\beta=0$ for smaller values of $\pi_G$. However, for $\gamma=0.3$, see Fig.~\ref{fig:expo_xxdot}(c), we find Wardrop equilibria ($\beta^*(\alpha)=\alpha$) in the interior of $[0,\beta_{\max}]$. Further numerical exploration \fdp{confirmed that the equilibria form an interval. Thus, again, we find that there exist  conditions (in this case, smaller $\Lpu$) when the system has a continuum of equilibria as in previous cases. }

\begin{figure}[t]
  \centering
      \includegraphics[width=0.30\textwidth]{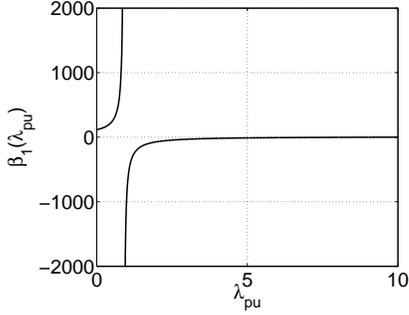}      
  \caption{The shape of function $\beta_1(\Lpu)$ for increasing values of $\Lpu$: the vertical asymptote corresponds to the value $\Lpu^s$.}\label{fig:betauno}
\end{figure}

\section{Users with side information}\label{sec:sideinfo}

\fdp{In the previous section we have considered the product of the trend and magnitude of the viewcount as 
a metric: as seen there, the structure of the equilibria that we can expect resembles closely what we found in the previous 
cases: either extremal Wardrop equilibria or a continuum of restpoints. We want to describe  the case when potential viewers 
 may be provided additional information on the upcoming popularity of a certain content, e.g., relying on some predictors or some apriori information they have. 
They judge whether to access or not a given content based on the product of the popularity $X$ and the 
popularity trend $\dot X$. But, they only know how such metric is going to accumulate over time, i.e., the metric for a user that 
approaches the content at time $t$ is }
\begin{equation}
y(t)=\int_{t}^{\tau} X(u) \dot X(u) du = \frac 12 (X^2(\tau) - X^2(t))  \nonumber
\end{equation}
\fdp{This metric can be used as a simple benchmark case: it contains information on the future dynamics of $X(\theta)$, 
 and it is defined by the current and the final values of the viewcount. However, the amount of such information in general is not 
sufficient at time $t$ to state the type of the content. Of course, more sophisticated metrics are possible. Nevertheless,
 the one at hand will do for the purpose of showing that by making the potential viewers of a 
 content aware of some side information, the system may experience a deep change in the structure 
 of the equilibria.}
 
Let all users adopt strategy 
\[
\alpha := y(t_\alpha), \quad 0 \leq t_\alpha \leq \tau 
\]
and in the same way as done before we want to determine the best response for a user adopting 
$\beta := y(t_\beta)$ as a reply, where  $0 \leq t_\beta \leq \tau$.

In the case $\beta \geq \alpha$, we recall that the dynamics is 
\[
X(t,\theta)=\alpha + \lambda(\theta)(t-t_{\alpha})
\]
where $\L(\theta):=(\Lpu + \Lps(\theta))$ for the sake of notation, so that 
\[
\alpha+\L(\theta)(t_{\beta}-t_{\alpha})=\sqrt{X^2(\tau,\theta)-2\beta}
\]
which solves for $\displaystyle t_{\beta}=\frac 1{\L(\theta)} \Big ( \alpha \frac{\Lpu}{\Lps(\theta)} + \sqrt{X^2(\tau,\theta)-2\beta} \Big)$.
The corresponding expression for the utility is $ U(\alpha,\beta) =$
\begin{eqnarray}
U_0(\alpha,\beta)-\left [ \frac{\pi_G \sqrt{X^2(\tau,G)-2\beta} }{\L(G)}- \frac{\pi_B \sqrt{(X^2(\tau,B)-2\beta)}}{\L(B)} \right ] \nonumber
\end{eqnarray}
where the term $U_0(\alpha,\beta)=(\pi_G-\pi_B)\tau - \alpha \Lpu \Big (  \frac{\pi_G}{\Lps(G)\L(G)}- \frac{\pi_B}{\Lps(B)\L(B)} \Big )$ 
 and it turns out that 
\[
\frac {dU(\alpha,\beta)}{d\beta} =\frac{\pi_G }{\L(G)(X^2(\tau,G)-2\beta)^{\frac 12} }- \frac{\pi_B }{\L(B)(X^2(\tau,B)-2\beta)^{\frac 12}}
\]
which is \fdp{decreasing with $\beta \in [-\infty,\beta_{\tau,B}]$, where $\beta_{\tau,B}:=\frac 12 X(\tau,B)$ as follows by comparing the ratio of the two positive terms appearing in the expression above under the assumption $X(\tau,G) \geq X(\tau,B)$)}. When $\frac{\pi_G }{\L(G)}\not = \frac{\pi_B }{\L(B)}$ the $U(\cdot,\beta)$ over $\mathbb R$ attains a unique maximum at
\[
\beta_1=\frac 12 \frac{- X^2(\tau,G)\big ( \frac{\pi_B}{\L(B)} \big )^2 + X^2(\tau,B)\big ( \frac{\pi_G}{\L(G)} \big )^2 }{\big ( \frac{\pi_G}{\L(G)} \big )^2 - \big ( \frac{\pi_B}{\L(B)} \big )^2}
\]
so that there exists also one maximum of $U(\alpha,\beta)$ in $[t_\alpha,\tau]$. 

We can distinguish three cases based on the fact that 
\begin{enumerate}
\item $\beta_1\leq \alpha$: the best response in this case is $\beta^*(\alpha)=\alpha$
\item $\alpha<\beta_1<\beta_{\tau,B}$: the best response is $\beta^*(\alpha)=\beta_1$
\item $\beta_1\geq \beta_{\tau,B}$: the best response in this case is $\beta^*(\alpha)=\beta_{\tau,B}$.
\end{enumerate}
\fdp{Finally, we notice that when $\frac{\pi_G }{\L(G)} = \frac{\pi_B }{\L(B)}$, case 1) applies.} 

\fdp{In the case $\beta < \alpha$, we can derive a similar analysis starting from the dynamics $X(t,\theta)=\Lps(\theta) t$, so that}
\[
\beta=y(t_\beta)=\frac 12 \Big ( X^2(\tau,\theta) - \Lps^2(\theta) t_\beta^2 \Big )
\]
so that $t_\beta=\sqrt{X^2(\tau,\theta)-2\beta}$, and 
\begin{eqnarray}
&&U(\alpha,\beta)=(\pi_G-\pi_B)\tau \nonumber \\
&&- \left [ \frac{\pi_G \sqrt{X^2(\tau,G)-2\beta} }{\Lps(G)}- \frac{\pi_B \sqrt{(X^2(\tau,B)-2\beta)}}{\Lps(B)} \right ] \nonumber
\end{eqnarray}
\fdp{In turn, we can recognize} the same structure for the best response as in the previous case, where the maximum of $U(\cdot,\beta)$ (when $\frac{\pi_G }{\Lps(G)}\not= \frac{\pi_B }{\Lps(B)}$), over $\mathbb R$ is attained at
\[
\beta_2=\frac 12 \frac{- X^2(\tau,G)\big ( \frac{\pi_B}{\Lps(B)} \big )^2+ X^2(\tau,B)\big ( \frac{\pi_G}{\Lps(G)} \big )^2 }{\big ( \frac{\pi_G}{\Lps(G)} \big )^2 - \big ( \frac{\pi_B}{\Lps(B)} \big )^2}
\]
and the three cases write
\begin{enumerate}
\item $\beta_2\leq 0$: the best response in this case is $\beta^*(\alpha)=0$.
\item $0 < \beta_2 < \alpha$: the best response is $\beta^*(\alpha)=\beta_2$.
\item $\beta_2 \geq \alpha$: the best response is $\beta^*(\alpha)=\alpha$.
\end{enumerate}
Again,  when $\frac{\pi_G }{\Lps(G)}= \frac{\pi_B }{\Lps(B)}$, case 1) applies. 

\fdp{Now, to complete our analysis, we need to determine the best response between the two cases: we need to detail the relation between $\beta_1$ and $\beta_2$. To so do we can rewrite for the sake of convenience }
\[
\beta_1(x)=\frac 12 \frac{\pi_G^2 x^2 X^2(\tau,B) - \pi_B^2 (L+x)^2 X^2(\tau,G)}{\pi_B^2 x^2 - \pi_G^2 (L+x)^2}
\]
where $L=\Lps(G)-\Lps(B)$ and $x=\Lps(G)+\Lpu$. It can be easily showed that 
\[
\frac d{dx} \beta_1(x)=\pi_G^2\pi_B^2\frac{2Lx (X^2(\tau,G)-X^2(\tau,B))(x-L)}{(\pi_B^2 x^2 - \pi_G (L+x)^2)^2}
\]
which brings $\frac d{dx} \beta_1(x)>0$ for $x\geq 0$, with a singularity in 
\[
\Lpu^s=\frac{\pi_B}{\pi_G - \pi_B}(\Lps(G)-\Lps(B))-\Lps(B)
\]
The typical shape of $\beta_1$ is reported in Fig.~\ref{fig:betauno}.
We observe that $\beta_1(\Lpu=0)=\beta_2$. The asymptotic value for $\Lpu=\infty$ is
\[
\beta_1(\infty)=\frac 12 \frac{\pi_G^2X^2(\tau,B) - \pi_B^2X^2(\tau,G)}{\pi_G^2 - \pi_B^2}
\]
\fdp{It can be verified that $\beta_1(\Lpu)$ is injective. Hence, the above analysis let us state: $\beta_1(\infty) \leq \beta_1(0) = \beta_2$, which in turn leads to the following}
\begin{lemma}
For $0\leq \Lpu < \Lpu^s$, it holds $\beta_1 \geq \beta_2$, and for $\Lpu > \Lpu^s$ 
it holds $\beta_1 < \beta_2$.
\end{lemma}

Now we can combine the conditions above to derive:
\begin{thm}\label{thm:threshold}
Let $I=[0,\beta_{\tau,B}]$\\
{\noindent i.} If $\Lpu > \Lpu^s$, then 
\[
W_s=[\beta_1,\beta_2]\cap I
\]
is the set of symmetric Wardrop equilibria for the system. 

{\noindent ii.} If $\Lpu < \Lpu^s$ then $W_s \subseteq \{0, \beta_{\tau,B}\}$. 

\end{thm}
\begin{proof}
\fdp{Case i. follows immediately observing that for $\beta_1\leq \alpha$ the best response is $\beta^*(\alpha)=\alpha$
and for $\beta_2 \geq \alpha$ the best response is $\beta^*(\alpha)=\alpha$: both conditions are satisfied simultaneously 
for $\alpha \geq 0$ if and only if $\alpha \in W_s$.\\
Case ii. is proved observing that the conditions for case i fail, so that only extremal cases can hold. In particular, $W_s$ is 
not always the empty set: if $\beta_2 \geq 0$, then $\beta_1\leq 0$ so that $\beta^*(0)=0$ and the same holds in the opposite case, i.e., if  $\beta_1 \geq \alpha=\beta_{\tau,B}$ then $\beta_1 \geq \alpha=\beta_{\tau,B}$ so 
that $\beta^*(\beta_{\tau,B})=\beta_{\tau,B}$.}\end{proof}



\fdp{The result in Thm.~\ref{thm:threshold} let us observe a neat phase transition effect on $\Lpu$: when the intensity of the views due to the pull mechanism is below threshold $\Lpu^s$, only extremal Wardrop equilibria are possible. Above that threshold, there can exist a continuum of equilibria where the system can settle.  
Let $\mu(\cdot)$ denote the standard real measure: a sufficient condition is provided in the following
\begin{cor}
$\mu(W_s)>0$ if  $\Lpu>\Lpu^s$ and $\beta_2\geq 0 >\beta_1$.  
\end{cor}
We can observe that $\pi_G < \pi_B$ implies $\beta_2\geq 0$ and $\Lpu>0>\Lpu^s$, 
so that a stronger sufficient condition than the one just provided in turn becomes: $\pi_G < \pi_B$ and $\beta_1\leq \beta_{\tau,B}$.
}

\section{Related Works}\label{sec:rel}


The analysis of dynamics of popularity of online contents has been subject of 
recent papers. The work \cite{Gill07youtubetraffic} provides an analysis of the YouTube system,
with comprehensive view of the characteristic of the generated traffic. 

In \cite{ChatFirstStep} the authors address the relation between metrics used 
to evaluate popularity. They observed that viewcount is strongly correlated with 
 several such metrics as number of comments, ratings, or favorites. However, all 
such metrics do not correlate to average rating. In this paper we confine our analysis 
to viewcount as the metric of interest. \cite{SzaboPop} focuses on the core problem of predicting 
popularity, namely, the viewcount, based on early measurements of user access. Based 
on YouTube videos or Digg stories measurements, the authors observe that contents 
increasing fast their viewcount in early stages typically become popular later on. 
The proposed empirical model, i.e., $\log N(t_r)=\log N(t_o)+\lambda_0(t_r,t_0)$ where $\lambda_0(t_r,t_0)$,  
is a random multiplicative noise and $N(t_r),N(t_0)$ is the viewcount at $t_r$ and $t_0$; it
resembles closely the exponential model adopted in this work. 

In \cite{RatkiewiczBurstyPoP} the authors propose a model accounting for change of ranking 
induced by UGC online platforms. The model is meant to overcome the limitations of the preferential 
attachment models. Those models in fact cannot explain bursty growth of content popularity; 
 those in turn are claimed an inherent property of the online platforms. The authors relate 
 bursty growth spikes to the way such systems expose popular contents to users and perform 
re-ranking of existing contents causing positive feedback loops.

The paper \cite{ChaTON} provides analysis of power law  behavior for the rank distribution 
of contents; the distribution of most watched videos is found heavily skewed towards 
the most popular ones.

Threshold models similar to those studied in this work are described by Granovetter~\cite{GranovetterAJS1978} 
in social science. The assumption is that individuals make binary decisions (in our framework, view or not view a content), 
according to some static internal threshold of others participating. A generalization based on threshold 
distribution is addressed in \cite{RolfeSocNet}.

\section{Conclusions}\label{sec:concl}

In this paper we characterized the access to online contents 
by game theoretical means by leveraging on the concept of Wardrop equilibrium. We deduced the 
structure of equilibria in systems where users adopt threshold type policies to select online contents. We explored several cases: 
 the case when the plain viewcount is the metric, or the viewcount trend, or both are combined as a product metric. We explored 
the case of a fixed time horizon dictated by the content lifetime, and we considered a case when 
the time horizon is not fixed. Finally we explored the impact of side information available 
to users. 

In all such cases we deduced the presence of a continuum of equilibria, which has potential implications in the 
design and control of platforms for online content access. In future work, in particular, we are exploring the dynamics associated 
to such sets of interior restpoints, when they exist, and comparing those with typical dynamics of online contents. 
However, not only equilibria are relevant: as showed in \cite{GraSoongJEBO1986}, threshold strategies,  
under specific conditions, may well lead the system to be asymptotically unstable; system trajectories may in turn consist of cycles that 
can move into a chaotic dynamics, essentially indistinguishable from random noise. 


\bibliographystyle{IEEEtran}	
\bibliography{utube}		

\end{document}